\documentclass[11pt,onecolumn]{IEEEtran}

\usepackage{pgf,tikz,pgfplots}
\usepackage{algorithm}
\usepackage{algpseudocode}

\usepackage{enumitem}
\usepackage{hyperref}
\hypersetup{
 colorlinks,
 linkcolor={blue!100!black},
 citecolor={blue!100!black},
 urlcolor={blue!80!black}
}

\usepackage{amsmath,amssymb,amsfonts,amsthm,mathtools}

\usepackage[numbers,sort]{natbib}

\usepackage{multirow}

\newtheorem{theorem}{Theorem}
\newtheorem{remark}{Remark}

\newtheorem{proposition}{Proposition}
\newtheorem{example}{Example}
\newtheorem{lemma}{Lemma}
\newtheorem{definition}{Definition}

\newcommand{\bF}{\mathbb{F}}

\newcommand{\bR}{\mathbb{R}}

\newcommand{\cA}{\mathcal{A}}
\newcommand{\cB}{\mathcal{B}}
\newcommand{\cC}{\mathcal{C}}
\newcommand{\cD}{\mathcal{D}}

\newcommand{\cK}{\mathcal{K}}
\newcommand{\cL}{\mathcal{L}}

\newcommand{\cR}{\mathcal{R}}
\newcommand{\cS}{\mathcal{S}}

\newcommand{\cV}{\mathcal{V}}

\newcommand{\boldb}{\mathbf{b}}
\newcommand{\boldc}{\mathbf{c}}

\newcommand{\bolde}{\mathbf{e}}

\newcommand{\boldi}{\mathbf{i}}
\newcommand{\boldj}{\mathbf{j}}
\newcommand{\boldk}{\mathbf{k}}
\newcommand{\boldl}{\mathbf{l}}
\newcommand{\boldm}{\mathbf{m}}

\newcommand{\boldo}{\mathbf{o}}

\newcommand{\boldr}{\mathbf{r}}
\newcommand{\bolds}{\mathbf{s}}

\newcommand{\boldv}{\mathbf{v}}
\newcommand{\boldw}{\mathbf{w}}
\newcommand{\boldx}{\mathbf{x}}
\newcommand{\boldy}{\mathbf{y}}

\DeclareMathOperator{\FtwoSpan}{\operatorname{span}_{\bF_2}}

\DeclarePairedDelimiter{\floor}{\lfloor}{\rfloor} % Use \floor* to scale.
\DeclarePairedDelimiter{\ceil}{\lceil}{\rceil} % Use \ceil* to scale.
\DeclareSymbolFont{bbold}{U}{bbold}{m}{n}
\DeclareSymbolFontAlphabet{\mathbbold}{bbold}
\newcommand{\1}{\mathbbold{1}}

\title{Linear Codes for Hyperdimensional Computing}

\author{\textbf{Netanel Raviv}\\
	\normalsize 
	Department of Computer Science and Engineering, Washington University in St. Louis, St. Louis, MO}

\begin{document}
\maketitle
\thispagestyle{empty}
\begin{abstract}
	Hyperdimensional Computing (HDC) is an emerging computational paradigm for representing compositional information as high-dimensional vectors, and has a promising potential in applications ranging from machine learning to neuromorphic computing. One of the long-standing challenges in HDC is factoring a compositional representation to its constituent factors, also known as \textit{the recovery problem}. In this paper we take a novel approach to solve the recovery problem, and propose the use of \textit{random linear codes}. These codes are subspaces over the Boolean field, and are a well-studied topic in information theory with various applications in digital communication. We begin by showing that hyperdimensional encoding using random linear codes retains favorable properties of the prevalent (ordinary) random codes, and hence HD representations using the two methods have comparable information storage capabilities. We proceed to show that random linear codes offer a rich subcode structure that can be used to form key-value stores, which encapsulate most use cases of HDC. Most importantly, we show that under the framework we develop, random linear codes admit simple recovery algorithms to factor (either bundled or bound) compositional representations. The former relies on constructing certain linear equation systems over the Boolean field, the solution to which reduces the search space dramatically and strictly outperforms exhaustive search in many cases. The latter employs the subspace structure of these codes to achieve provably correct factorization. Both methods are strictly faster than the state-of-the-art resonator networks, often by an order of magnitude. We implemented our techniques in Python using a benchmark software library, and demonstrate promising experimental results.
\end{abstract}

\section{Introduction}
Hyperdimensional Computing (HDC, also known as Vector Symbolic Architecture) is an emerging computational paradigm which aims to blur the ubiquitous location-based representation of information, where meaning of bits is tied to their relative position~\cite{kanerva2014computing}. Unlike traditional computing, in which data structures are constructed using context-based positioning of objects (e.g., lists, trees, vectors, etc.), the HDC paradigm identifies atomic items as high-dimensional vectors of context-free bits, which are then used as building blocks to form more sophisticated data structures---again high-dimensional (HD) vectors---through various algebraic operations. The HDC framework is motivated by a vast spectrum of scientific challenges and engineering applications, ranging from explanatory models of cognition \cite{plate2003holographic,gayler2004vector,lake2017building}, to fundamental tasks in machine learning \cite{yu2022understanding,imani2018hierarchical,thomas2022streaming,rahimi2016robust}, and to various applications such as DNA sequencing \cite{imani2018hdna}, speech recognition \cite{imani2017voicehd}, machine learning hardware \cite{schmuck2019hardware}, robotics \cite{neubert2019introduction}, and more. HDC attracted significant research efforts of late, with several theoretical studies \cite{thomas2021theoretical,clarkson2023capacity} and surveys \cite{kleyko2023survey,aygun2023learning,chang2023recent}.

The HDC paradigm is further divided into different \textit{architectures}, each of which specifies an \textit{alphabet}, a \textit{bundling} operator, and a \textit{binding} operator. The choice of alphabet specifies the entries of the atomic HD vectors; common choices are ``dense'' (i.e.,~$\{\pm1\}$), ``sparse'' (i.e.~$\{0,1\}$), or in some cases complex numbers~\cite{gallant2013representing,yu2022understanding} (see also~\cite{plate1995holographic}). A bundling operator receives as input two HD vectors and outputs an HD vector similar to both; common choices are the real-valued addition \cite{thomas2021theoretical,clarkson2023capacity}, Boolean OR \cite{yu2022understanding}, or the the majority operator \cite{schmuck2019hardware}. A binding operator receives as input two HD vectors and outputs an HD vector dissimilar to both. By and large, all works in HDC specify binding as point-wise (or Hadamard) product, which in the case of dense (i.e.,~$\{\pm1\}$) representation is equivalent to exclusive OR. 

Clearly, the choice of a specific architecture strongly depends on the specification of the system in which HDC is implemented. In this paper we focus on the most prevalent architecture, commonly referred to as Multiply-Add-Permute (MAP, although our methods obviate the need for the permutation operation), in which the HD vectors are dense, the bundling operator is real-valued addition~``$+$'', and the binding operator is exclusive-OR~``$\oplus$''.	

Various computational challenges arise in the context of the system specification in which HDC is employed, among which are recall (i.e., identifying if a given object is stored in memory), comparison (i.e., how ``similar'' are two given representations), or training on HD representations~\cite{thomas2022streaming}. Recall and comparison were theoretically analyzed by~\cite{thomas2021theoretical}, and shown to be strongly correlated with a property called \textit{incoherence}, which encapsulates the extent to which different HD representations are dissimilar, and stems from quasi-orthogonality~\cite{kainen1993quasiorthogonal}. One of the most important open problems in HDC, that is the first open problem mentioned in the encompassing survey of~\cite{kleyko2023survey}, is that of \textit{recovery}. In various applications one is given a representation of some complex data structure as input, and needs to factor it to its constituent elements; the input could be the result of bundling multiple HD vectors (i.e., \textit{bundling-recovery}) and/or binding such vectors (i.e., \textit{binding-recovery}). Recovery is a crucial ingredient in some important HDC systems such as general key-value stores (which instantiate most other data structure as special cases), or more domain specific applications such as visual scene analysis or search tree queries~\cite{frady2020resonator}. Several heuristics have been proposed to handle the recovery problem~\cite{frady2020resonator,kleyko2023efficient}, and yet, no rigorous approach is known as of yet.

With our main motivation being efficient recovery algorithms, in this paper we propose the use of \textit{random linear codes} in HDC. The term ``linear codes'' refers to subspaces over the Boolean field with large pairwise Hamming distance. These objects are well studied in information and coding theory, mainly for error correction in communication systems~\cite{roth2006introduction}. Choosing linear codes at random is attractive due to the high probability of obtaining a code with favorable properties, such as large Hamming distance. 

We begin by showing that essentially \textit{nothing is lost} by specifying the HDC encoding process to random linear codes. That is, the information capacity of HDC systems, as measured by~\cite{thomas2021theoretical} using the incoherence property mentioned earlier, is not adversely affected by imposing a linear structure on the encoding process. In fact, random linear codes admit a \textit{better} bound than nonlinear ones on the probability of choosing a suitable code. Furthermore, encoding with any linear code can be done using simple logical operations (specifically, exclusive OR and AND), can be implemented in software/hardware with ease, and requires exponentially less space to store. Further, linear codes can also support choosing codewords at random on-the-fly, which is central to the HDC paradigm, with marginal additional effort. 

Then, we proceed to show that \textit{much is gained} by the use of linear codes in HDC, specifically, in three aspects:
\begin{enumerate}
	\item Simple implementation of key-value stores, which provides a unified approach for studying and implementing many HD data structures.
	\item Drastic reduction of the size of the search space in the bundling-recovery problem.
	\item Efficient and provably correct solution to the binding-recovery problem.
\end{enumerate}
Our bundling-recovery algorithm relies on exhaustive search over a certain subset of codewords, albeit exponentially smaller than the size of the codebook in its entirety. Our binding-recovery algorithm, however, does not rely on exhaustive search, nor on any type of heuristics, and in fact, applies to all linear codes regardless of their incoherence properties. It relies on the simple linear-algebraic fact that a vector in a subspace has a unique representation over a basis of that subspace, which can be found in linear time by solving a system of linear equations. We implemented both algorithms using off-the-shelf Python libraries, and showed promising results. 

Specifically, our bundling-recovery algorithm shows drastic reduction of run-time compared to exhaustive search, and surprisingly, also consistently succeeds in cases where exhaustive search fails (e.g., due to insufficient incoherence). Additionally, our binding-recovery algorithm offers an almost-instantaneous and correct solution to all linear codes, and is far superior to a recent benchmark implementation~\cite{heddes2023torchhd} of a technique called \textit{resonator networks}~\cite{frady2020resonator,kent2020resonator}. For example, our algorithm normally concludes within less than a second, whereas resonator networks often require several minutes only to create a list of all codewords. We also comment that due to the specific way we use linear codes, combinations of bundled and bound vectors can also be recovered. For example, to recover all codewords in a bundle of bound vectors one can first apply our bundling-recovery algorithm, and then apply the binding-recovery algorithm on each of the resulting bound vectors.

This paper is structured as follows. In Section~\ref{section:preliminaries} we provide a brief introduction to HDC, including the analysis of~\cite{thomas2021theoretical} for the connection between information capacity and incoherence. Also in Section~\ref{section:preliminaries}, we provide a primer on Boolean linear algebra and linear codes, and iterate the (known) fact that random linear codes have high information capacity (i.e., low-coherence), with a better probability bound than not-necessarily-linear random codes. In Section~\ref{section:KV} it is shown how the subcode structure of linear codes can accommodate a simple implementation of key-value stores, which instantiate most use-cases of HDC as special cases. To illustrate the strengths of our unified approach for building key-value stores, we provide several examples for special cases where our approach specifies to finite sets, vectors (sequences), search trees, and visual scene analysis. In Section~\ref{section:recovery} we present our bundling- and binding-recovery algorithms, and formally prove their guarantees. Experiments are provided in Section~\ref{section:experiments}, and conclusions are drawn in Section~\ref{section:discussion}.

\section{Preliminaries}\label{section:preliminaries}

We begin with an introduction to HDC in Section~\ref{section:preliminariesHDC}, where we present a theoretical perspective to HDC and focus on the recovery problem. We proceed with a primer on linear codes, in which we familiarize the reader with Boolean algebra and coding theory, and present several similarities between HDC and coding theory.

\subsection{Hyperdimensional Computing}\label{section:preliminariesHDC}
The Hyperdimensional Computing (HDC) paradigm replaces conventional representation of data (e.g., vectors, sets, trees, etc.) by the use of high-dimensional random vectors of length~$n$ (typically $1000\le n\le 10000$, although smaller values are also common). Each participating element, e.g., a pixel in an image, a word in a language, etc., is assigned an \textit{atomic} random vector of length~$n$. The use of random vectors is a key component of HDC, as random vectors present a \textit{quasi-orthogonal} property, in the sense that due to concentration of measures the inner product between random~$\{\pm1\}$-vectors is highly likely to be close to zero~\cite{kainen1993quasiorthogonal}.  

More complex data structures, such as sentences or images, are obtained via algebraic manipulations of atomic vectors. Many frameworks were put forward to facilitate these manipulations, of which we focus on the most common Multiply, Add, Permute (MAP)~\cite{gayler2004vector}. By and large all such frameworks contain a ``binding'' and ``bundling'' operations. In the MAP framework the bundling operation, used to produce a vector similar to both input vectors, is implemented using the ordinary addition of integers, denoted by~`$+$'. Binding operation is implemented by point-wise (or Hadamard) product, denoted by~`$\oplus$'. 

For example, in order to represent a set containing the elements~$e_1,\ldots,e_s$, taken from some arbitrary domain, one encodes each~$e_i$ to an HD vector~$\bolde_i$ using a randomly generated codebook~$\cC$, and bundles all resulting HD vectors together as~$\bolds=\sum_{i=1}^s\bolde_s$. Later, to perform \textit{recall}, i.e., to determine if a given element~$e$ belongs to the set~$\{e_i\}_{i=1}^s$, one uses the same codebook~$\cC$ to encode~$e$ to an HD vector~$\bolde$, and decides according to the value of~$\bolde\bolds^\intercal$. It was shown in~\cite{thomas2021theoretical} that successful recall depends on the \textit{incoherence} of the codebook~$\cC$. 

\begin{definition}\label{definition:muIncoherentCodes}
	For~$\mu>0$, a codebook~$\cC\subseteq \{\pm1\}^n$ is called $\mu$-incoherent if~$\boldc\boldc'^\intercal\le \mu n$ for every distinct~$\boldc,\boldc'\in\cC$.
\end{definition}

Furthermore, if one uses a~$\mu$-incoherent code~$\cC$ with~$\mu<\tfrac{1}{2s}$, then the truth value of~``$e\in\{e_i\}_{i=1}^s$'' is identical to the truth value of~``$\bolde\bolds^\intercal\ge n/2$'' \cite[Thm.~2]{thomas2021theoretical}.

More complex data structures than sets can be given a useful HD representation using the binding operation~$\oplus$. For example, key-value (KV) stores are data structures which contain pairs of keys taken from a set~$K$, and values taken from a set~$V$ (e.g., a phonebook is a KV-store which contains pairs of the form~$(\text{name, \text{number}})$). An HD representation of KV-stores requires two codebooks~$\cK,\cV$ to map every key~$k\in K$ to an HD vector~$\boldk\in\cK$, and every value~$v\in V$ to an HD vector $\boldv\in \cV$. A KV-store~$\{(k_i,v_i)\}_{i=1}^s$ is mapped to the HD vector~$\bolds=\sum_{i=1}^s\boldk_i\oplus\boldv_i$. 

A KV-store should support queries of the form ``\textit{what is the value of key~$k$?}''. To answer this query in HDC one finds the element~$v\in V$ whose respective~$\boldv\in \cV$ maximizes the value of~$\boldv(\bolds\oplus\boldk)^\intercal$. It was shown in~\cite[Thm~.14]{thomas2021theoretical} that successful answer to such queries is guaranteed if some ``joint'' incoherence of the codebooks~$\cK,\cV$ holds, i.e., if~$\boldv(\boldk\oplus\boldv')^\intercal\le\mu n$ for every~$\boldv,\boldv'\in\cV$ and every~$\boldk\in\cK$. Our proposition of using random linear codes for HDC obviates this separate treatment of coherence by considering~$\cK,\cV$ as \textit{subcodes} of a $\mu$-incoherent linear code~$\cC$, as will be explained in detail in the sequel.

HDC often requires the use of \textit{recovery} algorithms, which factor a complex representation to its constituent atomic vectors. Arguably the simplest use case is the aforementioned query to a KV-store, where one seeks to recover the elements of~$\bolds\oplus\boldk$ in search of a unique~$\boldv\in\cV$ that is the answer to the query. Another use-case appears in the context of machine learning, where a neural-network is trained to map a complex image to a complex HD vector which represents the identity and location of objects in the image (e.g., \cite[Sec.~4.2]{frady2020resonator}); recovery algorithms are then required to interpret the output of the neural network.

Formally, recovery algorithms can be partitioned to two types, \textit{bundling-recovery} and \textit{binding-recovery}, or $\Sigma$-recovery and~$\oplus$-recovery, respectively, for short. In~$\Sigma$-recovery one is given a bundle of HD vectors~$\bolds=\sum_{i=1}^s\bolde_i$ and wishes to obtain the~$\bolde_i$'s; the~$\bolde_i$'s might be complex HD vectors themselves (rather than atomic). In~$\oplus$-recovery one is given a bound HD vector~$\bolds=\bigoplus_{i=1}^s\boldc_i$, where each~$\boldc_i$ is taken from a different codebook~$\cC_i$, and wishes to obtain the~$\boldc_i$'s. 

A common approach toward~$\Sigma$-recovery of a bundle~$\bolds=\sum_{i=1}^s\bolde_i$ is \textit{exhaustive search}. As mentioned earlier, if the codebook~$\cC$ is sufficiently incoherent, one can traverse all~$\boldc\in \cC$ and decide if~$\boldc\in\{\bolde_i\}_{i=1}^s$ according to the value of~$\bolds\boldc^\intercal$. However, this requires~$|\cC|$ many inner products, and does not scale well with large codebooks. We emphasize that while the number of elements~$s$ is usually small (e.g., the types of animals in a picture), the size of~$\cC$ might be very large (e.g., all types of animals on earth). 

The state-of-the-art~$\oplus$-recovery method for a bound vector~$\bolds=\bigoplus_{i=1}^s\boldc_i$ is \textit{resonator-networks}~\cite{frady2020resonator,kent2020resonator}, which can be seen as inspired by Hopfield networks and associative memories. Although some variants exist \cite{kent2020resonator,hersche2023decoding}, the main gist is as follows. A resonator network maintains an approximate vector~$\hat{\boldc}_i(t)$ for each~$\boldc_i$ and each step~$t$ of the algorithm, which is initialized as~$\hat{\boldc}_i(0)=\sum_{\boldc\in\cC_i}\boldc$. In a series of computation steps, each~$\hat{\boldc}_i(t)$ is updated as~$\hat{\boldc}_i(t+1)=\operatorname{sign}(C_iC_i^\intercal (\bolds\oplus \bigoplus_{j\in[s]\setminus\{i\}}\hat{\boldc}_j(t))^\intercal)$, where~$C_i$ is a matrix whose rows are the HD vectors in~$\cC_i$. The rationale is that~$C_iC_i^\intercal$ is a projection matrix on the subspace spanned by the codebook~$\cC_i$ over~$\bR$, and binding the approximations~$\hat{\boldc}_j(t)$ for~$j\ne i$ with~$\bolds$ cancels the effect of~$\boldc_j$ on~$\bolds$ and isolates the required~$\boldc_i$. Resonator networks demonstrate experimental success, and yet the author is not aware of any theoretical convergence guarantees. 

To better elucidate the contribution of this paper, the following remarks are in place. 
\begin{remark}[Exact vs. approximate] 
	We focus on exact HD computations, where the answer to a certain computation must be correct in all cases, rather than with high probability. For instance, in the recall case (see above), we are interested in constructing algorithms which determine the truth value of~``$e\in\{e_i\}_{i=1}^s$'' in all cases (for a certain parameter regime). Attaining approximate algorithms is left for future work. 
\end{remark}	

\begin{remark}
	To prevent ambiguity between coding-theory and HDC notions, we set the following conventions. A \emph{code} is a subset~$\cC\subseteq\{\pm1\}^n$ (rather than ``codebook''). A vector~$\boldv\in\{\pm1\}^n$ is called \emph{binary} (rather than ``bipolar''), and the values~$\{\pm1\}$ are called \emph{bits}. An element~$\boldc\in\cC$ of a code is called a \emph{codeword}, and an element~$\boldw\in\{\pm1\}^n$, which does not necessarily belong to any spcified code, is called a \emph{vector} (or an HD vector, if we wish to emphasize its high dimension~$n$). Some sources (e.g., \cite{thomas2021theoretical}) identify a code~$\cC$ with an encoding function which maps elements to codewords. It will be clear shortly that in linear codes, such encoding functions can be easily and efficiently implemented using Boolean-algebraic operations, and hence we omit the encoding function from the discussion. %, given any Boolean representation of elements. 
\end{remark}

\subsection{A primer on linear codes}\label{section:preliminariesCodes}
We provide a short introduction to the theory of error correcting codes; for further information the interested reader is referred to~\cite{roth2006introduction}. Linear codes are defined via Boolean algebra, i.e., algebraic operations over the binary field. Traditionally, the binary field~$\bF_2$ contains two abstract elements~``$0$'' and~``$1$'', but here we represent those elements using the real numbers~$1$ and~$-1$, where the real~$-1$ represents the Boolean~``$1$'' and the real~$1$ represents the Boolean~``$0$'' (i.e.~$\text{real}=(-1)^{\text{Boolean}}$). The Boolean field has an ``addition'' operation\footnote{The choice of the notation~`$\oplus$' here deliberately coincides with the operator~$\oplus$ used for the binding operation in Section~\ref{section:preliminariesHDC} for reasons that will be made clear in the sequel.}~$\oplus$ and a ``multiplication'' operation~$\odot$, which are defined using the Boolean functions excluside OR (XOR) and AND, respectively, see Table~\ref{table:addition} and Table~\ref{table:multiplication}. Notice that using~$\{\pm1\}$ to represent the Boolean field enables Boolean addition (i.e., XOR) to be implemented using multiplication over~$\bR$.

\begin{table}[ht]
	\begin{minipage}[b]{0.45\linewidth}
		\centering
		\begin{tabular}{|c|c|c|}
			\hline
			$\oplus$ & $1$ & $-1$ \\
			\hline
			$1$ & $1$ & $-1$ \\
			\hline
			$-1$ & $-1$ & $1$ \\
			\hline
		\end{tabular}
		\vspace{0.3cm}
		\caption{Addition over~$\bF_2$.}\label{table:addition}
	\end{minipage}
	\hspace{0.5cm}
	\begin{minipage}[b]{0.45\linewidth}
		\centering
		\begin{tabular}{|c|c|c|}
			\hline
			$\odot$ & $1$ & $-1$ \\
			\hline
			$1$ & $1$ & $1$ \\
			\hline
			$-1$ & $1$ & $-1$ \\
			\hline
		\end{tabular}
	\vspace{0.3cm}
		\caption{Multiplication over~$\bF_2$.}\label{table:multiplication}
	\end{minipage}
\end{table}

Once the structure of the Boolean field~$\bF_2$ is established, one can define algebraic notions in~$\bF_2^n$, i.e., the vector space of dimension~$n$ over~$\bF_2$. That is, one can freely consider matrices, subspaces, linear transformation etc., which have (nearly) identical properties to their~$\bR^n$ counterparts. For instance, the notation~$\odot$ can be extended to multiplication of matrices in a natural way, e.g., for~$\boldx,\boldy\in\bF_2^n$ let~$\boldx\odot\boldy^\intercal=\bigoplus_{i=1}^nx_i\odot y_i$. When discussing matrix multiplication, we omit the~$\odot$ notation if it is clear from the context that the product is defined over~$\bF_2$. Similarly, linear span is defined over~$\bF_2$, rather than over~$\bR$---for vectors~$\boldc_1,\ldots,\boldc_k$ we let~$\FtwoSpan\{\boldc_1,\ldots,\boldc_k\}$ be their linear span with~$\bF_2$ coefficients, i.e., the set~$\{\alpha_1\odot\boldc_1\oplus\ldots\oplus\alpha_k\odot\boldc_k\vert \alpha_i\in\bF_2\}$, or equivalently, the set~$\{ \bigoplus_{j\in J}\boldc_j\vert J\subseteq [k] \}$, where~$[k]$ denotes~$\{1,2,\ldots,k\}$. Linear independence of vectors also follows naturally from the definition of linear span.

In coding theoretic jargon, a (binary) \textit{code} is a subset of~$\bF_2^n$. A code consists of Boolean vectors of length~$n$ called \textit{codewords}. A code is called \textit{linear} if it is a subspace of~$\bF_2^n$. The dimension~$\dim_{\bF_2}(\cC)$ of a code~$\cC$ is usually denoted in the literature by~$k$, and the size~$|\cC|$ of a $k$-dimensional binary code~$\cC$ is~$2^k$. A linear code which is a~$k$-dimensional subspace of~$\bF_2^n$ is said to be an~``$[n,k]_2$ code.'' %As a linear subspace, 

An~$[n,k]_2$ code~$\cC$ can be defined in two equivalent ways. One, using a \textit{generator matrix}~$G\in\bF_2^{k\times n}$, whose rows are a basis for~$\cC$, i.e.,~$\cC=\{\boldx \odot G\vert \boldx\in\bF_2^k\}$. Two, using a \textit{parity check matrix}~$H\in\bF_2^{(n-k)\times n}$, whose right kernel is~$\cC$, i.e.,~$\cC=\ker H=\{\boldc\vert H\odot \boldc^\intercal=\1\}$, where~$\1\in\bF_2^{n-k}$ is the ``zero vector,'' i.e., contains only\footnote{Recall that we use the~$\pm1$ representation of~$\bF_2$, and hence the Boolean~``$0$'' is represented by~$1$. In this spirit, we say that a vector~$\boldy$ is ``nonzero'' if~$\boldy\ne \1$.}~$1$'s. Generator matrices provide a convenient encoding algorithm, where a Boolean string~$\boldx\in\bF_2^k$, which represents information to be encoded, is encoded to~$\boldx G\in\bF_2^n$. Parity-check matrices provide a convenient way to \textit{test} if a given vector~$\boldc$ belongs to a given code~$\cC$ by the truth value of the expression~``$H\boldc^\intercal=\1$.'' %Both generator and parity-check matrices are used in the sequel.

\begin{example}\label{example:linearcodefordummies}
	The following code~$\cC$ is a~$[5,2]_2$ code, i.e., a code of length~$5$, dimension~$\dim_{\bF_2}(\cC)=2$, and size~$2^{\dim_{\bF_2}(\cC)}=4$.
	\begin{align*}
		\cC=\{(1,1,1,1,1),(-1,-1,1,1,1),(-1,1,-1,-1,-1),(1,-1,-1,-1,-1)\}.
	\end{align*}
	A generator matrix of~$\cC$ is given by
	\begin{align*}
		G=
		\begin{bmatrix}
			-1&-1&\phantom{-}1&\phantom{-}1&\phantom{-}1\\
			-1&\phantom{-}1&-1&-1&-1
		\end{bmatrix}\in\bF_2^{2\times 5},
	\end{align*}
	whose rows are an~$\bF_2$-basis (of size~$\dim_{\bF_2}(\cC)=2$) to~$\cC$. It can be verified that the following matrix is a parity-check matrix for~$\cC$, i.e.~$\cC=\{\boldc\in\bF_2^5\vert H\odot \boldc^\intercal=\1\}$.
	\begin{align*}
		H = 
		\begin{bmatrix}
			-1&-1&\phantom{-}1&\phantom{-}1&-1\\
			-1&-1&\phantom{-}1&-1&\phantom{-}1\\
			-1&-1&-1&\phantom{-}1&\phantom{-}1
		\end{bmatrix}\in\bF_2^{3\times 5}
	\end{align*}
\end{example}

Since linear codes are subspaces, they contain subspaces themselves, and these are called \textit{subcodes}. In what follows, for a linear code~$\cC$ we write~$\cC=\cK\times\cV$ if~$\cK$ and~$\cV$ are subcodes (i.e., subspaces) of~$\cC$ and every~$\boldc\in\cC$ can be written uniquely as~$\boldc=\boldk\oplus\boldv$ with~$\boldk\in\cK$ and~$\boldv\in\cV$ (i.e., a direct sum); such subcodes~$\cK$ and~$\cV$ must also intersect trivially (i.e.,~$\cK\cap \cV=\{\1\}$). 

\begin{example}
	In Example~\ref{example:linearcodefordummies}, we have that
	\begin{align*}
		\cC_1=\{ (1,1,1,1,1),(-1,-1,1,1,1) \}\mbox{ and }
		\cC_2=\{  (1,1,1,1,1),(-1,1,-1,-1,-1)\}
	\end{align*}
	are subcodes of~$\cC$. It can be verified that~$\dim_{\bF_2}(\cC_1)=\dim_{\bF_2}(\cC_2)=1$ and that~$\cC=\cC_1\times\cC_2$. More broadly, for positive integers~$\ell<k$, suppose~$\cC=\FtwoSpan\{\boldc_1,\ldots,\boldc_k\}$ for some vectors~$\boldc_1,\ldots,\boldc_k$ which are linearly independent over~$\bF_2$. Then, the subcodes~$\cK=\FtwoSpan\{ \boldc_1,\ldots,\boldc_{\ell} \}$ and~$\cV= \FtwoSpan\{ \boldc_{\ell+1},\ldots,\boldc_{k} \}$ satisfy~$\cC=\cK\times \cV$.
\end{example}

In addition, recall the notion of \textit{affine subspaces}, which are subspaces that are shifted by a vector. In the context of linear codes, such affine subspaces are called \textit{cosets}. For instance, for a Boolean vector~$\boldy$ the set~$\cC\oplus \boldy=\{\boldc\oplus\boldy\vert \boldc\in\cC\}$ is a coset of~$\cC$ (and also an affine subspace of~$\bF_2^n$). It is known that~$\cC\oplus\boldy=\cC$ if and only if~$\boldy\in\cC$, that different cosets are disjoint, and that the union of all cosets of a code~$\cC$ is the entire space~$\bF_2^n$. Furthermore, the following property of subcodes and cosets will be useful in the sequel for constructing HD-representations of key-value stores.

\begin{lemma}\label{lemma:notinV}
	Let~$\cC=\cK\times \cV$,~$\boldv\in\cV$, and~$\boldk\in \cK$. If~$\boldk\ne \1$, then~$\boldk\oplus\boldv\notin\cV$.
\end{lemma}
\begin{proof}
	Assume for contradiction that~$\boldk\oplus\boldv\in\cV$, which implies that~$\boldk=\boldv\oplus\boldv'$ for some~$\boldv'\in\cV$. If~$\boldv=\boldv'$ then~$\boldk=\1$, a contradiction, and hence~$\boldv\ne\boldv'$. Therefore,~$\boldk$ is a nonzero vector which lies in~$\cK\cap \cV=\{\1\}$, also a contradiction.
\end{proof}

An important property of codes is their \textit{minimum Hamming distance} (minimum distance, for short). The minimum distance of a code~$\cC$ is the minimum value of~$d_H(\boldc,\boldc')$, where~$d_H$ denotes Hamming distance, and where~$\boldc,\boldc'$ range over all pairs of distinct codewords of~$\cC$. A well-known property of linear codes, which is essential in our analysis, is that their minimum distance is equal to their \textit{minimum Hamming weight} (minimum weight, for short). That is, for a linear code~$\cC$ we have $\min_{\boldc\in\cC\setminus\{\1\}}w_H(\boldc)=\min_{\boldc,\boldc'\in\cC,\boldc\ne\boldc'}d_H(\boldc,\boldc')$, where~$w_H$ denotes Hamming weight (i.e., the Hamming distance from~$\1$, or equivalently, the number of~$-1$ entries). Similarly, a linear code's maximum distance is equal to its maximum weight. For instance, in Example~\ref{example:linearcodefordummies} the minimum weight of~$\cC$ is~$2$, and the maximum weight is~$4$; by checking the Hamming distance between every pair it can be verified that the minimum distance of~$\cC$ equals its minimum weight.

\paragraph{$\mu$-incoherent codes and $\epsilon$-balanced codes} As mentioned in Section~\ref{section:preliminariesHDC}, the essence of successful HDC is using~$\mu$-incoherent codes (Definition~\ref{definition:muIncoherentCodes}), and more crucially, the fact that random (not-necessarily-linear) codes are~$\mu$-incoherent with high probability. We now demonstrate a simple fact, that~$\mu$-incoherent codes are equivalent to \textit{$\epsilon$-balanced codes}; $\epsilon$-balanced \textit{linear} codes are a well studied notion~\cite{naor1990small,alon1992simple,ben2009constructing,ta2017explicit}, with a wide range of applications in computer science.

\begin{definition}\label{definition:epsBalanced}
	For~$\epsilon>0$, a code~$\cC$ is called~$\epsilon$-balanced if $(1/2-\epsilon)n\le d_H(\boldc,\boldc')\le (1/2+\epsilon)n$ for every distinct~$\boldc,\boldc'\in\cC$. 
\end{definition}
Due to the aforementioned equivalence between minimum distance and minimum weight, as well as between maximum distance and maximum weight, for linear codes Definition~\ref{definition:epsBalanced} reduces to~$(1/2-\epsilon)n\le w_H(\boldc)\le (1/2+\epsilon)n$ for every nonzero~$\boldc\in\cC$. Since we use the~$\pm1$ representation of~$\bF_2$, we have that~$d_H(\boldx,\boldy)=\frac{1}{2}(n-\boldx\boldy^\intercal)$, and thus the following is immediate.

\begin{proposition}\label{proposition:incoherentBalanced}
	For any~$\mu>0$, a code is~$\mu$-incoherent if and only if it is~$\frac{\mu}{2}$-balanced.
\end{proposition}

Constructing explicit~$\epsilon$-balanced linear codes is notoriously difficult~\cite{naor1990small,alon1992simple,ben2009constructing,ta2017explicit}. However, choosing a uniformly random~$[n,k]_2$ linear code (by selecting a generator matrix uniformly at random) is likely to result in an~$\epsilon$-balanced code with~$\epsilon=\Theta(\sqrt{k/n})$. While this is a well-known fact, for the lack of an accessible reference we provide the full details below.

\paragraph{Random linear codes are~$\epsilon$-balanced with high probability}

Our first key message in this paper is that nothing is lost in HDC when employing random-linear, rather than just random, codes. Furthermore, one can provide a better bound for the success in choosing a suitable linear code, than the corresponding bound for not-necessarily-linear ones.

Formally, we quantify the above statement according to the resulting incoherence. A quick exercise, which requires applying Hoeffding's inequality and the union-bound and provided in~\cite[Sec.~3.1.2]{thomas2021theoretical}, shows that a random code~$\cA$ of size~$m$ and length~$n$ is~$\epsilon$-balanced with high probability (i.e., which tends to zero as the size of the code tends to infinity) for~$\epsilon=\Theta(\sqrt{\ln m/n})$. More precisely, for~$\cA$ a (not-necessarily-linear) code whose~$m$ codewords are chosen uniformly at random, we have that
\begin{align*}
	\Pr(\cA\text{ is not $\epsilon$-balanced})\le m^2e^{-2\epsilon^2n}.
\end{align*}
Notably, the~$m^2$ term comes from using the union bound over all pairs of codewords in~$\cA$. When bounding the probability of a random \textit{linear}~$[n,k]_2$ code~$\cC$ to be~$\epsilon$-balanced, however, one only needs to bound the Hamming weight of individual words. Specifically, for a linear~$[n,k]_2$ code ~$\cC$ (with~$m=2^k$ codewords) whose generator matrix~$G$ is chosen uniformly at random we have that
\begin{align*}
	\Pr(\text{$\cC$ is not $\epsilon$-balanced})&=\Pr(\exists \boldx\in\bF_2^k\setminus\{\1\}, |w_H(\boldx G)-n/2|>\epsilon n)\\
	&\le \sum_{\boldx\in\bF_2^k\setminus\{\1\}}\Pr(|w_H(\boldx G)-n/2|>\epsilon n)\\
	&\le(2^k-1)\cdot e^{-2\epsilon^2n}\le m e^{-2\epsilon^2n}.
\end{align*}
where the penultimate transition follows since~$G$ is a uniformly random matrix, hence~$\boldx G$ is a uniformly random vector whose Hamming weight expectation is~$n/2$, and thus Hoeffding's inequality is applicable. Notice that the dependence on the code size~$m$ is now only linear, rather than quadratic as in the nonlinear case. 

\noindent\textbf{Our proposition.} We propose using random linear codes for HDC. That is, we fix a desired length~$n$ and choose a generator matrix~$G\in\bF_2^{k\times n}$ uniformly at random. Then, whenever a codeword is to be chosen for a given element, choose~$\boldx\in\bF_2^k$ according to some specified rule (e.g., according to the natural representation of the element in bits), and let~$\boldx G$ be the desired codeword. Also, it is straightforward to adapt this framework to choose~$G$ on-the-fly rather than all at once a priori, and the details are discussed in the appendix. The above discussion shows that nothing is lost in terms of information content by specifying the HDC framework to random linear codes. The remainder of this paper is dedicated to demonstrating the benefits of linear codes in HDC, in the form of a unified treatment for KV-stores and recovery algorithms.

Finally, we mention that our use of the term ``decoding'' adheres to the meaning of this notion in coding theory (unlike, say,~\cite{thomas2021theoretical}), i.e., the process of recovering~$\boldx$ from (a potentially noisy version of)~$\boldx G$. Decoding without noise can be easily done by solving a linear equation over~$\bF_2$.  However, decoding random linear codes under noise is NP-hard~\cite{berlekamp1978inherent}, even though some algorithms with small exponential coefficients are known~(e.g.,~\cite{becker2012decoding}), as well as deep learning based methods for decoding linear codes~(e.g.,~\cite{nachmani2016learning}), and practically efficient decoding methods for restricted classes such as LDPC (e.g.,~\cite{chen2005reduced}).

\section{Key-value stores from linear codes.}\label{section:KV}
We now turn to show how linear codes can be used to implement hyperdimensional KV-stores. Some of the proposed techniques require the use of~$\oplus$-recovery and~$\Sigma$-recovery algorithms; for now we refer to those as black-boxes, and present them explicitly in the sequel.

A key-value (KV) store is a fundamental data structure which stores entries (values) with indices (keys), and should answer queries of the form ``\textit{what is the value for a given key?}''. In this section it is shown that linear encoding for keys and values provides a framework for maintaining KV-stores using HD vectors. We begin by presenting the basic routines which must be supported by KV-stores, showing how to implement them in practice using linear codes, proving correctness, and analyzing complexity. 

Later, we show that our framework can be extended seamlessly, again using subcode structure, to the case where either the keys or the values (or both) have a \textit{compositional} structure, i.e., where a given key (or value) can be further decomposed to atomic parts of some contextual meaning.
An example where such compositional structure is required is in the implementation of a search tree as an HD vector; since the keys represent a path in said tree, they can be further decomposed to a sequence of integers which represents that path. Finally, to show the generality of our framework we demonstrate its use in storing sets and vectors (sequences), as well as visual scene analysis and search trees.

Formally, a KV-store~$\cS$ is a collection of pairs of the form~$(k,v)$, where the keys~$k$ are taken from some set~$K$ and the values~$v$ from some set~$V$, and where stored keys must be distinct from one another\footnote{The variant which allows identical keys can be seen as a key-set (KS) store, where a given key corresponds to a set of elements. Our framework can be extended to KS-stores as well.} (values are allowed to be identical). The size of~$\cS$, i.e., the number of pairs in it, is denoted by~$s$. A KV-store must support the following routines:
\begin{itemize}
	\item $\texttt{Add}(k,v)$ should add~$(k,v)$ to~$\cS$ if there is no pair of the form~$(k,\cdot)$ in~$\cS$, and otherwise should return a ``failure'' symbol~$\bot$.
	\item $\texttt{Remove}(k)$ should remove~$(k,v)$ from~$\cS$ if exists, and otherwise should return~$\bot$.
	\item $\texttt{Return}(k)$ should return~$v$ if $(k,v)\in \cS$, and~$\bot$ otherwise.
\end{itemize}
We comment that one often requires \textit{bidirectional} KV-stores, i.e., which additionally support a~$\texttt{Return}(v)$ for a value~$v\in V$. Our framework can be seamlessly extended to bidirectional KV-stores as well, and yet we focus on KV-stores as stated above for better clarity.

To construct a hyperdimensional vector~$\bolds$ for a KV-store~$\cS$ with key set~$K$ and value set~$V$, we consider an~$\epsilon$-balanced code~$\cC$ (see Section~\ref{section:preliminariesCodes}) which has two subcodes~$\cK,\cV$ such that~$\cC=\cK\times \cV$. Additionally,  we require that $|K|\le|\cK|$ and~$|V|\le|\cV|$, i.e., the subcode~$\cK$ (resp.~$\cV$) has at least as many codewords as possible keys (resp. values). We let~$\kappa:K\to\cK$ be an encoding function which maps keys in~$K$ to codewords of~$\cK$, and let~$\nu:V\to\cV$ be an encoding function which maps values in~$V$ to codewords of~$\cV$. We initialize the HD-representation of~$\cS$ as~$\bolds=0$, and the above functions are implemented as follows.
\begin{itemize}
	\item $\texttt{Add}(k,v)$: If~$\texttt{Return}(k)=\bot$, perform~$\bolds=\bolds+\kappa(k)\oplus\nu(v)$, and otherwise return~$\bot$.
	\item $\texttt{Remove}(k)$: If $\texttt{Return}(k)=v$ perform~$\bolds=\bolds-\kappa(k)\oplus \nu(v)$, and otherwise return~$\bot$.
	\item $\texttt{Return}(k)$: There are two options with different complexities.
	\begin{enumerate}[start=1,label={(\Alph*)}]
		\item\label{item:ReturnMax} Let~$p=\max_{v\in V}[\bolds\oplus \kappa(k)]\nu(v)^\intercal$. If~$p>n(\epsilon+1/2)$, return the maximizing value (i.e., return~$\arg\max_{v\in V}[\bolds\oplus \kappa(k)]\nu(v)^\intercal$), and otherwise return~$\bot$.
		\item\label{item:ReturnSigma} Perform~$\Sigma$-recovery (Section~\ref{section:SigmaRecovery}) over~$\bolds\oplus\kappa(k)$, and let~$\boldr_1,\ldots,\boldr_s$ be the resulting vectors. Let~$H_\cV$ be the parity check matrix of~$\cV$. Return the unique~$\boldr\in\{\boldr_i\}_{i=1}^s$ such that~$H_{\cV}\boldr^\intercal=\1$ if exists, and otherwise return~$\bot$.
	\end{enumerate}
\end{itemize}

Since the correctness and complexity of $\texttt{Add}(k,v)$ and $\texttt{Remove}(k)$ follow immediately from that of $\texttt{Return}(k)$, we focus on the latter.

\begin{theorem}\label{theorem:returnCnC}
	Let~$s$ be an upper bound on the number of pairs in the KV-store~$\cS$.
	If~$\epsilon<\frac{1}{4s-2}$, then Option~\ref{item:ReturnMax} implements~$\texttt{Return}(k)$ correctly, and has complexity~$O(n|V|)$, and Option~\ref{item:ReturnSigma} implements~$\texttt{Return}(k)$ correctly (subject to the limitations of the underlying~$\Sigma$-recovery algorithm), and has complexity~$O(C+s(n-\log|V|)n)$, where the underlying~$\Sigma$-recovery algorithm runs in~$O(C)$.
\end{theorem}

\begin{proof}
	Denote~$\cS=\{(k_i,v_i)\}_{i=1}^s$, which by the operations $\texttt{Add}(k,v)$ and $\texttt{Remove}(k)$ implies that~$\bolds=\sum_{i=1}^s\boldk_i\oplus\boldv_i$, where~$\boldk_i\triangleq\kappa(k_i)$ is the encoding of key~$k_i$ in~$\cK$ and~$\boldv_i\triangleq \nu(v_i)$ is the encoding of value~$v_i$ in~$\cV$.   		
	
	\noindent{Option~\ref{item:ReturnMax}:}
	The complexity is~$O(n|V|)$, since~$2n$ multiplications and~$n-1$ additions (over~$\bR$) are needed for every~$v\in V$ in the~$\max$ operation. To show correctness, observe that 
	\begin{align*}
		\bolds\oplus\kappa(k)=\sum_{i\in[s]}\boldk_i\oplus\kappa(k)\oplus \boldv_i.
	\end{align*}
	and split to two cases.
	
	\underline{Case 1:} There exists~$j\in[s]$ such that~$k=k_j$. We have that
	\begin{align}\label{equation:ReturnMax}
		\bolds\oplus\kappa(k)&=\bolds\oplus\boldk_j=\boldv_j+\sum_{i\in[s]\setminus\{j\}}\boldk_i\oplus \boldk_j\oplus\boldv_i\nonumber\\
		&\triangleq\boldv_j+\sum_{i\in[s]\setminus\{j\}} \boldk'_i\oplus\boldv_i,
	\end{align}
	for some vectors~$\boldk'_i\triangleq\boldk_i\oplus\boldk_j,i\in[s]\setminus\{j\}$. We note that the vectors~$\{ \boldk_i' \}_{[s]\setminus\{j\}}$ all lie in~$\cK\setminus \{\1\}$ since~$\boldk_1,\ldots,\boldk_s$ are distinct, and since~$\cK$ is closed under~$\oplus$. Now, recall that since~$\cC=\cK\times \cV$, it follows that~$\boldk\oplus\boldv\in\cC$ for every~$\boldk\in \cK$ and every~$\boldv\in \cV$, namely, that binding an encoded key with an encoded value always produces an element of the underlying~$\epsilon$-balanced code~$\cC$. Therefore, all elements in the summation in~\eqref{equation:ReturnMax} are codewords of~$\cC$. Moreover, it follows from Lemma~\ref{lemma:notinV} that since~$\boldk_i'\ne\1$ for all~$i\in[s]\setminus\{j\}$, we have that~$\boldk'_i\oplus\boldv_i\notin\cV$, and in particular~$\boldk'_i\oplus\boldv_i\ne \boldv_j$. Therefore, since~$\cC$ is~$\epsilon$-balanced, it follows from Proposition~\ref{proposition:incoherentBalanced} that it is~$2\epsilon$-incoherent (i.e.,~$|\boldc'\boldc^\intercal|\le2\epsilon n$ for every distinct~$\boldc',\boldc\in\cC$), and therefore~\eqref{equation:ReturnMax} implies that
	\begin{align}\label{equation:isj}
		[\bolds\oplus\kappa(k)]\boldv_j^\intercal\ge n - (s-1)\cdot 2\epsilon n,
	\end{align}
	whereas for every~$\ell\in[s]\setminus\{j\}$ we have
	\begin{align}\label{equation:notj}
		[\bolds\oplus\kappa(k)]\boldv_\ell^\intercal&=\boldv_j\boldv_\ell^\intercal+\sum_{i\in[s]\setminus\{j\}}(\boldk'_i\oplus\boldv_i)\boldv_\ell^\intercal\nonumber\\ &\le 2\epsilon n +\sum_{i\in[s]\setminus\{j\}}(\boldk'_i\oplus\boldv_i)\boldv_\ell^\intercal\quad \text{(since $\ell\ne j$)}\nonumber\\
		&\le 2\epsilon n+(s-1)\cdot 2\epsilon n\quad\text{(since~$\boldk'_i\oplus\boldv_i\ne \boldv_\ell$ for all~$i$, while both are in~$\cC$)}\nonumber\\
		&=2s\epsilon n.
	\end{align}
	By observing that~$n-(s-1)\cdot 2\epsilon n > 2s\epsilon n$ if and only if~$\epsilon<\frac{1}{4s-2}$ (which is given), and that~$(n-(s-1)\cdot2\epsilon n+2s\epsilon n)/2=n(\epsilon+1/2)$, Eqs.~\eqref{equation:isj} and~\eqref{equation:notj} imply that
	\begin{align*}
		[\bolds\oplus\kappa(k)]\boldv_j^\intercal>n(\epsilon+1/2)>[\bolds\oplus\kappa(k)]\boldv_\ell^\intercal
	\end{align*}
	for all~$\ell\in[s]\setminus\{j\}$. Hence, 
	the maximum~$p$ is larger than~$n(\epsilon+1/2)$, and its corresponding maximizer is~$v_j$, which is indeed the value corresponding to the queried key~$k_j$.
	
	\underline{Case 2:} $k$ is neither of the stored keys~$k_1,\ldots,k_s$. Similar to~\eqref{equation:ReturnMax} we have that
	\begin{align}\label{equation:ReturnMax2}
		\bolds\oplus\kappa(k)=\bolds\oplus\boldk=\sum_{i\in[s]}\boldk_i\oplus \boldk\oplus\boldv_i,
	\end{align}
	where~$\boldk=\kappa(k)$. The fact that~$k\notin\{ k_i \}_{i=1}^s$ implies that~$\boldk_i\oplus\boldk\in\cK\setminus\{\1\}$ for all~$i\in[s]$, and hence~$\boldk_i\oplus\boldk\oplus\boldv_i\notin\cV$ (and in particular~$\boldk_i\oplus\boldk\oplus\boldv_i\ne \boldv_\ell$ for every~$\ell\in[s]$) for all~$i\in[s]$ by Lemma~\ref{lemma:notinV}. Therefore, for every~$\ell\in[s]$ we have
	\begin{align*}
		[\bolds\oplus\kappa(k)]\boldv_\ell^\intercal=[\bolds\oplus\boldk]\boldv_\ell^\intercal\overset{\eqref{equation:ReturnMax2}}{=}\sum_{i\in[s]}(\boldk_i\oplus\boldk\oplus\boldv_i)\boldv_\ell^\intercal\le s\cdot2\epsilon n,
	\end{align*}
	where the last inequality follows from the fact that~$\boldv_\ell$ and~$\boldk_i\oplus\boldk\oplus\boldv_i$ are distinct codewords in a~$2\epsilon$-incoherent code, for all~$i\in[s]$, and hence their inner product is at most~$2\epsilon n$. Now, observing that~$\epsilon<\frac{1}{4s-2}$ (which is given) implies that~$s\cdot2\epsilon n<n(\epsilon+1/2)$, hence the maximizer~$p$ is at most~$n(\epsilon+1/2)$, and we correctly return~$\bot$.
	
	\noindent Option~\ref{item:ReturnSigma}: Let~$O(C)$ be the complexity of the underlying~$\Sigma$-recovery algorithm, and recall that the parity check matrix~$H_\cV$ of~$\cV$ is an~$(n-\dim_{\bF_2}(\cV))\times n$ matrix over~$\bF_2$. On top of performing the~$\Sigma$-recovery algorithm we perform~$s$ parity-checks using~$H_\cV$, and hence the overall complexity is~$O(C+s(n-\dim_{\bF_2}\cV)n)$. Moreover, since we choose a code~$\cV$ which is large enough to contain all keys in~$V$, we have that~$|\cV|\ge|V|$, and since~$\dim_{\bF_2}\cV=\log|\cV|$ it follows that the complexity is~$O(C+s(n-\log|V|)n)$.
	
	To show correctness, assuming that the underlying~$\Sigma$-recovery algorithm recovers the factors~$\boldr_1,\ldots,\boldr_s$ of~$\bolds\oplus\kappa(k)$ correctly, it suffices to show that: $(a)$ if~$k=k_j$ for some~$j\in[s]$ then~$H_\cV\odot\boldr^\intercal=\1$ for a unique~$\boldr\in\{ \boldr_i \}_{i=1}^s$, and that~$\boldr=\boldv_j$; and~$(b)$ if~$k\ne k_j$ for all~$j\in[s]$ we have that~$H_\cV\odot\boldr_\ell^\intercal\ne \1$ for all~$\ell\in[s]$. 
	
	To show~$(a)$, assume that~$k=k_j$ for some~$j\in[s]$. Then, by~\eqref{equation:ReturnMax} we have that
	\begin{align*}
		\bolds\oplus\kappa(k)=\boldv_j+\sum_{i\in[s]\setminus\{j\}} \boldk'_i\oplus\boldv_i,
	\end{align*}
	which implies that~$\{\boldr_i\}_{i=1}^s=\{\boldv_j\}\cup\{\boldk_i'\oplus\boldv_i\}_{i\in[s]\setminus\{j\}}$. Indeed, since~$\boldk_i'=\boldk_i\oplus\boldk_j\ne \1$ for all~$i\in[s]\setminus\{j\}$, we have that~$\boldk_i'\oplus\boldv_i\notin\cV$ by Lemma~\ref{lemma:notinV}, which implies that~$H_\cV\odot (\boldk_i'\oplus\boldv_i)^\intercal\ne \1$ for all~$i\in[s]\setminus\{j\}$; further, since~$\boldv_j\in\cV$ we have that~$H_\cV\odot\boldv_j^\intercal=\1$ as required. Showing~$(b)$ is similar, since if~$k\ne k_j$ for all~$j\in[s]$ it follows from~\eqref{equation:ReturnMax2} that
	\begin{align*}
		\bolds\oplus\kappa(k)=\bolds\oplus\boldk=\sum_{i\in[s]}\boldk_i\oplus \boldk\oplus\boldv_i,
	\end{align*}
	where~$\boldk=\kappa(k)$, which implies that~$\{\boldr_i\}_{i=1}^s=\{ \boldk_i\oplus\bold \boldk\oplus \boldv_i\}_{i=1}^s$. Indeed, for all~$i\in[s]$, since~$\boldk_i\ne\boldk$ it follows that~$\boldk_i\oplus\boldk\oplus\boldv_i\notin\cV$ by Lemma~\ref{lemma:notinV}, and hence~$H_\cV\odot \boldr_i^\intercal\ne \1$, which concludes the proof.
\end{proof}

It is evident from Theorem~\ref{theorem:returnCnC} that the determining factor in choosing between Option~\ref{item:ReturnMax} and Option~\ref{item:ReturnSigma} is the size of the key space~$|V|$ relative to the complexity of executing a~$\Sigma$-recovery algorithm alongside~$s$ parity checks. 

\subsection{Compositional keys and values}\label{section:compositional}
In many applications in hyperdimensional computing one requires a set of keys or a set of values which has a compositional structure. For instance, in a KV-store which represents a search tree, a key~$\boldk$ represents a path in the tree, and a value~$\boldv$ represents the content of a leaf at the end of that path (see Section~\ref{section:examples} which follows for a detailed example). A path in this case is compositional in the sense that it describes which of the~$q$ siblings to follow (in a $q$-ary tree) in order to reach said leaf from the root of the tree. As another example, consider the case of visual scene analysis (also discussed in detail in  Section~\ref{section:examples}), i.e., where each ``scene'' is a ``key store'' (i.e., a KV-store with trivial value set~$V$ which contains only one element): each scene is a bundle of multiple \textit{objects} in the scene, and each object is a binding of multiple \textit{attributes} (such as type, color, position, etc.). Each such key (or object) is therefore compositional in the sense that it is obtained by binding multiple HD vectors. If an HD representation of a scene is to be factorized, one would first apply~$\Sigma$-recovery to obtain a list of all individual objects in the scene, and further apply~$\oplus$-recovery on each object to obtain its color, position, etc.

Due to this duality between keys and values in the context of~$\oplus$-recovery, in the remainder of this section we refer to either of them as \textit{labels}. Compositional labels can be supported using a straightforward combination of KV-stores (see previous section) with the~$\oplus$-recovery algorithm described in Section~\ref{section:oplusRecovery} to follow. Specifically, suppose a label set~$L$ has a compositional structure, i.e., it is contained in the product of several smaller sets~$L\subseteq L_1\times \ldots\times L_r$ for some~$r$. To impose this compositional structure on the frameworks described above, identify~$L$ with a subcode~$\cL$ of~$\cC$ (i.e.,~$\cL$ is either of the subcodes~$\cK$ or~$\cV$ mentioned in Section~\ref{section:KV}), and further decompose~$\cL$ to~$r$ subcodes~$\cL=\cL_1\times\ldots\times \cL_r$ (clearly, each~$\cL_i$ must be at least as large as~$L_i$), where each subcode~$L_i$ has a respective encoding function~$\lambda_i:L_i\to\cL_i$. 

For a given compositional label~$\ell=(\ell_1,\ldots,\ell_r)\in L$, encode it using a composition of the individual encoding functions:
\begin{align*}
	\lambda(\ell)=\lambda_1(\ell_1)\oplus\ldots\oplus\lambda_r(\ell_r).
\end{align*}
The operations \texttt{Return}, \texttt{Add}, and \texttt{Remove} work exactly as described earlier, by replacing the proper encoding function (i.e., either~$\kappa$ for keys or~$\nu$ for values) with~$\lambda$. If at any point a decomposition of~$\lambda(\ell)$ is required, apply the~$\oplus$-recovery algorithm of Section~\ref{section:oplusRecovery}, which is efficient, succeeds in all cases, and in fact, does not even rely on the incoherence property of~$\cC$.

\subsection{Examples}\label{section:examples}
In this section it is shown that the KV-store implementation above specifies, with a proper choice of an incoherent code~$\cC$ and its subcodes~$\cK$ and~$\cV$, to many use cases of HD data structures that were recently discussed in the literature. 

\subsubsection{Finite sets} \label{section:finitesets} Arguably the simplest data structure, a finite set must support addition and removal of items, as well as testing if a given item belongs to the set. To store a set of size at most~$s$ from a universe~$K$, follow the construction in Section~\ref{section:KV} as follows. Choose~$\epsilon<\frac{1}{4s-2}$ and an integer~$k$ such that each item in~$K$ can be represented using~$k$ bits---these bits could either be the item's index, a result of a compression algorithm, or be chosen at random, see Appendix~\ref{section:representation} for details. Then, choose~$n$ such that~$\epsilon\approx \sqrt{\frac{k}{n}}$, choose an~$[n,k]_2$ random linear code~$\cC$, and fix the (trivial) subcodes~$\cK=\cC$ and~$\cV=\{\1\}$ (which clearly satisfy~$\cC=\cK\times \cV$). 

For an item~$i\in K$, let~$\boldi$ be its representation using~$k$ bits, and define the encoding function~$\kappa:K\to\cK$ be~$\kappa(i)=\boldi G$, where~$G$ is the generator matrix of~$\cC$. Also, let~$\nu$ be the function which outputs~$\1$ for any input. 

To perform operations over the resulting data structure, begin by initializing~$\bolds=0$. To test if a given item~$i\in K$ is in~$\bolds$ apply~$\texttt{Return}(i)$; if the output is~$\1$ return~$\texttt{Yes}$, and if the output is~$\bot$ return~$\texttt{No}$. Given an item~$i\in K$ to store, apply the routine~$\texttt{Add}(i,1)$, which, if~$\texttt{Return}(i)=\texttt{No}$, adds~$\kappa(i)\oplus\nu(1)=\boldi G\oplus\1=\boldi G$ to~$\bolds$. Similarly, to remove an item~$i\in K$ first apply~$\texttt{Return}(i)$, and if the answer is~$\texttt{Yes}$, subtract~$\kappa(i)\oplus \nu(1)=\boldi G$ from~$\bolds$. Theorem~\ref{theorem:returnCnC} implies that these procedures implement a set correctly with the prescribed complexities.

\begin{remark}
	If one wishes not to choose the entire code~$\cC$ a priori, or if a representation of the elements of~$K$ as bits is not available, one can resort to random choices of codewords on the fly, with very light bookkeeping to maintain the linear structure of~$\cC$. The full details are given in the Appendix~\ref{section:onthefly}.
\end{remark}

\begin{remark}\label{remark:sameAsNonlinear}
	A set data structure built using a linear code is identical to one built using a not-necessarily-linear one in terms of their supposed incoherence, and hence their information content (i.e., the maximum set size~$s$ which guarantees successful remove, add, and membership operations). Therefore, any further property whose proof relies exclusively on the coherence of the underlying code is automatically inherited by our linear-code variant. For example, our methods of using linear codes to represent sets automatically inherit the size estimation, intersection computation, and noise resilience mentioned in~\cite[Thms. 8, 9, 10]{thomas2021theoretical}.
\end{remark}

\subsubsection{Vectors (sequences)} Arguably the second simplest data structure is one which stores \textit{ordered} elements. Specifically, the data structure should support the storage of at most~$s$ ordered elements from a universe~$V$, as well as inquiring about the content of entry~$i\in[s]$, replacing an item an entry~$i\in[s]$ by another, etc.

Traditionally (e.g., \cite{thomas2021theoretical,clarkson2023capacity,frady2018theory}), sequences are stored with the aid of an additional basic operation (on top of binding and bundling), the \textit{cyclic permutation}, as follows. To store a sequence of~$s$ elements~$(e_1,\ldots,e_s)$, where~$e_i\in V$ for all~$i$, each codeword for~$e_i$ is rotated cyclically~$i$ times (say, to the right), and the resulting vectors are bundled to form a vector~$\bolds$. Then, for instance, to return the element at position~$i$, the vector~$\bolds$ is rotated~$i$ times to left, and an ordinary recall algorithm is performed (e.g., by finding the maximizer of the inner product with~$\bolds$ over the entire codebook for~$V$). The reasoning behind this approach is that a random codeword is mutually independent (entry-wise) with its cyclic shifts. Therefore, after rotating~$i$ times to the left, only the correct codeword (i.e., which corresponds to the~$i$'th value in~$\bolds$) would have large inner-product with~$\bolds$, and the remaining rotated codewords would have small-inner products which would not interfere.

When using linear codes, however, similar intuitions can formalized and guaranteed \textit{without} the need for the additional permutation operation; this guarantees simpler implementation of sequential data structures using only binding and bundling. Specifically, choose~$\epsilon<\frac{1}{4s-2}$ and an integer~$k$ such that~$k=k'+\ceil{\log(s)}$, where every element in~$V$ has a binary representation using~$k'$ bits (see Appendix~\ref{section:representation}). Then, choose~$n$ such that~$\epsilon\approx \sqrt{\frac{k}{n}}$, choose an~$[n,k]_2$ random linear code~$\cC$, and choose subcodes~$\cK$ and~$\cV$ such that~$\dim_{\bF_2}\cK=\ceil{\log(s)}$, $\dim_{\bF_2}\cV=k'$, and~$\cC=\cK\times\cV$. The subcode~$\cK$ is used for indices, and we define its respective encoding accordingly---for an integer~$i\in[s]$ let~$\boldi$ be its binary (i.e.,~$\pm1$) representation using~$\ceil{\log(s)}$ bits, and let~$\kappa(i)=\boldi G_\cK$, where~$G_\cK$ is a generator matrix for~$\cK$. For~$\cV$ we follow similar encoding to the one we used for sets---for an element~$j\in V$ we let~$\boldj$ be its representation using~$k'$ bits, and let~$\nu(j)=\boldj G_\cV$, where~$G_\cV$ is a generator matrix for~$\cV$. 

Intuitively, in our method entry~$i$ of the resulting vector is encoded using codewords from the coset $\kappa(i)\oplus\cV$ of~$\cV$, which is contained in~$\cC$. Since different cosets are disjoint (see Section~\ref{section:preliminariesCodes}) and since they are all contained in~$\cC$, the inner product between codewords from different cosets will always be small. Formally, to add an element~$e\in V$ at position~$i$ apply~$\texttt{Add}(i,e)$; to remove the element from position~$i$ apply~$\texttt{Remove}(i)$; and to inquire about the content of entry~$i$ apply~$\texttt{Return}(i)$. Theorem~\ref{theorem:returnCnC} guarantees that these operations succeed.

\subsubsection{Search trees} To illustrate another strength of our proposed methods, we provide an example which uses the compositional capabilities of HDC with linear codes (Section~\ref{section:compositional}). We follow the use-case described in~\cite[Sec.~4.1]{frady2020resonator}, which describes an HD representation of a search tree. The work of~\cite{frady2020resonator} requires bundling, binding, and permuting, and a combination of~$\Sigma$-recovery and~$\oplus$-recovery for querying the resulting vector, which are implemented heuristically using resonator networks. In contrast, we describe herein a technique which only requires bundling and binding, and queries the resulting vector using~$\Sigma$-recovery and~$\oplus$-recovery, which can be efficiently implemented for linear codes as will be shown in Section~\ref{section:recovery}.

A search tree is seen as a collection of leaves. Each leaf is a tuple $(i_0,i_1,\ldots,i_\ell)$ with~$i_0$ the value in the leaf, and the path leading to it from the root is given by~$i_1,\ldots,i_\ell$, where~$i_j\in\{\texttt{left},\texttt{right}\}$ for~$j\in[\ell]$. In~\cite{frady2020resonator}, an HD vector for a search tree is constructed by bundling key-value pairs. Each pair is the binding of the leaf value~$i_0$ with a vector~$(i_1,\ldots,i_\ell)$ representing the path to it (the vector is encoded using cyclic permutations as described earlier). Our framework is different from~\cite{frady2020resonator} in a few key aspects. First, we support any number of offsprings for a given node in the tree. Second, we use the subcode structure to not only obviate the need to use permutations, but to eliminate the need to store ordered information altogether. Intuitively, to describe a path in a tree of depth~$\ell$ we use~$\ell$ subcodes, the~$i$'th of which describes which offspring to follow in step~$i$ of the path. Since subcodes are inherently distinct from one another, no further ordering information is required.

Formally, for parameters~$d,\ell,s,r$ with~$s\le (2^d-1)^\ell$, we wish to encode an $(2^d-1)$-ary search tree (i.e., at most~$2^d-1$ offsprings at each node) of depth at most~$\ell$ and at most~$s$ leaves, where each leaf contains a value which can be represented using~$r$ bits (see Appendix~\ref{section:representation}). Choose~$\epsilon<\frac{1}{4s-2}$ and an~$[n,k]_2$ code~$\cC$ such that~$\epsilon\approx \sqrt{\frac{k}{n}}$, where~$k=\ell d+r$, and choose subcodes~$\cK$ and~$\cV$ of~$\cC$ such that~$\dim_{\bF_2}\cK= r$, $\dim_{\bF_2}\cV= \ell d$, and~$\cC=\cK\times \cV$. 

Further, choose subcodes~$\cV_1,\ldots,\cV_\ell$ of~$\cV$ such that~$\dim_{\bF_2} \cV_i= d$ for all~$i\in[\ell]$, and~$\cV=\cV_1\times\ldots\times\cV_\ell$. As in previous examples, use an encoding function~$\kappa(i)=\boldi G_\cK$ to encode a value~$i\in K$ (i.e., labels in leaves), where~$\boldi$ is a representation of~$i$ using~$r$ bits and where~$G_\cK$ is a generator matrix for~$\cK$. In addition, for each~$j\in[d]$ and~$i\in[\ell]$ use~$\nu_i(j)=\boldj G_{\cV_i}$ to encode the fact that in depth~$i$ the path follows sibling~$j$ (or stops if~$j=0$), where~$\boldj$ is the~$d$-bit binary representation of~$j\in[\ell]$. As in Section~\ref{section:compositional} we use the~$\nu_i$'s to form a joint encoding function for~$\cV$---a path~$j=(j_1,\ldots,j_\ell)$ is encoded as~$\nu(j)=\nu_1(j_1)\oplus \ldots\oplus \nu_\ell(j_\ell)$.

The search tree data structure should support insertion of a given label at a given path, removal of a label if exists, and a \textit{path-query}\footnote{In~\cite{frady2020resonator} an additional \textit{leaf-query} in required, which returns a leaf-label given a path; this can be implemented using \textit{bi-directional} KV-stores, a slight extension of the KV-stores described in Section~\ref{section:KV}, and the details are omitted for brevity.}. A path-query receives a leaf-label~$i\in V$, returns the path leading to it, if exists, and otherwise returns~$\bot$. 

To support insertion, removal, and path-query using the above linear-code framework, use the following routines. To insert a leaf labeled~$i$ at path~$j=(j_1,\ldots,j_\ell)$ apply~$\texttt{Add}(i,j)$. To remove a leaf whose label is~$i$, apply~$\texttt{Remove}(i)$. To find what is the path, if exists, to a given label~$i\in V$ apply~$\texttt{Return}(i)$---if the output is~$\bot$ return~$\bot$, and otherwise, if the output is~$j$, use~$\oplus$-recovery over~$\nu(j)$ to decompose~$j$ to its constituent elements~$(j_1,\ldots,j_\ell)$. We emphasize again that our~$\oplus$-recovery algorithm is exact in this case. It returns the correct decomposition of~$j$ to~$(j_1,\ldots,j_\ell)$ in all cases, has linear running time, and only uses algebraic operations. This stands in stark contrast to the heuristic methods of resonator networks.

\begin{remark}
	The above routine for path-query is a special case of a more general approach. In the above we chose to represent the compositional structure of paths in a tree using the subcode structure of the (sub)code~$\cV$; we represent the path~$\nu(j)$ using the unique decomposition~$\nu(j)=\nu_1(j_1)\oplus\ldots\oplus\nu_\ell(j_\ell)$. However,~$\nu(j)$ is inherently connected to its \textit{decoding} (see Section~\ref{section:preliminariesCodes}), i.e., to the process of identifying the unique bit string~$\boldm$ of length $\dim_{\bF_2}\cV$ such that~$\nu(j)=\boldm G_\cV$. In turn,~$\boldm$ can represent a path~$j=(j_1,\ldots,j_\ell)$ in \emph{any} other arbitrary way (e.g., by enumeration). Therefore, one is free to use decoding, followed by any other interpretation of~$\boldm$ as a path.
\end{remark}

\subsubsection{Visual scene analysis} Another use-case which demonstrate the compositional capabilities of linear codes, is a data structure which describes a scene that contains superimposed objects, and each object contains multiple attributes. For example, imagine a picture in which multiple digits are overlayed on top of each other in different positions and different colors~\cite[Fig.~3]{frady2020resonator}. One would like a construct a single HD vector which describes such a scene, in a way which can be decoded efficiently, i.e., a given compositional HD vector can be decomposed to its constituent objects, and each object can be decomposed to its constituent attributes. A natural solution is using four codes to describe the color, the digit, the latitude and the longitude of the object. Then, a scene is a bundling of object vectors, each object vector is a binding of a codeword from each of these codes. We comment that this technique is useful in adding reasoning capabilities to machine learning models. By training these models to identify compositional structures, and then decoding those structures to their constituent elements, one improves the interpretation capabilities of these models exponentially.

Visual scenes are also a special case of our proposed KV-stores. Specifically, suppose~$s$ objects are to be superimposed, each object has~$d$ attributes (type, color, etc.), and each attribute can be represented using~$f$ bits. Choose~$\epsilon<\frac{1}{4s-2}$ and respective~$n$ and~$k$ such that~$\epsilon\approx \sqrt{\frac{k}{n}}$, where~$k=df$. Choose an~$[n,k]_2$ random linear code~$\cC$, and choose subcodes~$\cK=\cC$, $\cV=\1$, and~$\cC=\cK\times\cV$. Further, choose subcodes~$\cK_1,\ldots,\cK_d$ such that~$\dim_{\bF_2}\cK_i=f$ and~$\cK=\cK_1\times\ldots\times\cK_d$. As in previous examples, the subcodes~$\cK_i$ are endowed with encoding functions~$\kappa_i(j)=\boldj G_{\cK_i}$, where~$G_{\cK_i}$ is the generator matrix for~$\cK_i$, and~$\boldj$ is the representation of attribute~$j$ using~$f$ bits. These functions form the encoding function for~$\cK$ in its entirety, where~$i=(i_1,\ldots,i_d)$ is encoded as~$\kappa(i)=\kappa_1(i_1)\oplus\ldots\oplus\kappa_d(i_d)$, and finally, let~$\nu$ be the constant function which outputs~$\1$ for every input. 

The routines for implementing a visual scene are as follows. To add an element with attributes~$a=(a_1,\ldots,a_d)$ to the scene, perform~$\texttt{Add}(a,1)$, and to remove an element with such attributes perform~$\texttt{Remove}(a)$. To analyze a scene, i.e., to decompose it to its constituent elements, and then further decompose each element to its attributes, perform~$\Sigma$-recovery, followed by~$\oplus$-recovery for each object. 

We also comment that for successful application of this technique in machine learning (see above), it is essential to be able to handle noise. 
That is, a machine learning model which was trained to output compositional vectors is highly likely to make certain errors (e.g., bit flips, as a result of inaccurate inputs to the activation functions in the last layer of the model). To remedy such cases, one would have to apply an additional decoding algorithm to remove that noise prior to feeding it to our~$\Sigma$-recovery and~$\oplus$-recovery algorithms. 

\section{Recovery algorithms}\label{section:recovery}

\subsection{$\oplus$-recovery}\label{section:oplusRecovery}
The~$\oplus$-recovery problem is formally defined as follows~\cite[Sec.~2]{kent2020resonator}.
\begin{itemize}
	\item \textbf{Input}: $\boldc\in\{\pm1\}^n$, and codes~$\cC_1,\ldots,\cC_F$ with~$\cC_i\subseteq\{\pm1\}^n$ for all~$i\in[F]$.
	\item \textbf{Goal}: Find~$\boldc_1,\ldots,\boldc_F$ with~$\boldc_i\in\cC_i$ for all~$i$, such that~$\boldc=\bigoplus_{i=1}^F\boldc_i$.
\end{itemize}

In general, this is a hard problem, although no hardness results (e.g., NP-completeness) are known to the author. In~\cite{kent2020resonator,frady2020resonator} a heuristic solution using resonator networks is proposed. Yet, as we show shortly, if the~$\cC_i$'s are linear codes, then the problem can be solved via straightforward Boolean algebra.

\begin{remark}\label{remark:bases}
	Clearly, the complexity of solving~$\oplus$-recovery depends on the representation of the~$\cC_i$'s. In general, when the~$\cC_i$'s are completely random codes they have no structure which can be used for compression, hence providing~$\cC_i$ as input requires~$|\cC_i|\cdot n$ many bits, and the total input size is~$(\sum_{i=1}^F|\cC_i|+1)n$. Also, note that a simple exhaustive search would require about~$\prod_{i=1}^F|\cC_i|$ many computations, which is infeasible in most cases. In contrast, if the codes~$\cC_i$ are linear, only their generator matrices are required to represent them, which reduces the input size to~$(\sum_{i=1}^F\dim_{\bF_2}(\cC_i)+1)\cdot n=(\sum_{i=1}^F\log|\cC_i|+1)\cdot n$ bits.
\end{remark}

\begin{theorem}
	If~$\cC_1,\ldots,\cC_F$ are linear codes then the~$\oplus$-recovery problem is efficiently solvable.
\end{theorem}

\begin{proof}
	Let~$\cB_i\triangleq\{\boldb_{i,1},\ldots,\boldb_{i,b_i}\}\subseteq \cC_i$ be an~$\bF_2$-basis for~$\cC_i$, where~$b_i=\dim_{\bF_2}(\cC_i)$ for all~$i$, and let~$\cB'\subseteq \cup_{i=1}^F\cB_i$ be a maximal subset which is linearly independent over~$\bF_2$. 
	We first argue that~$\cB'$ is an~$\bF_2$-basis for the sum~$\cC\triangleq\FtwoSpan(\cup_{i=1}^F\cC_i)$, i.e., that~$\cC=\FtwoSpan(\cB')$, by showing bidirectional containment. The containment $\FtwoSpan(\cB')\subseteq\cC$ is clear since~$\cB'=\cup_{i=1}^F(\cB'\cap \cB_i)$, and each~$\cB'\cap\cB_i$ spans a subset of~$\cC_i$. To prove the containment $\cC\subseteq \FtwoSpan(\cB')$, assume for contradiction that there exists~$\boldc'\in\cC\setminus \FtwoSpan(\cB')$, which by definition can be written as~$\boldc'=\oplus_{j\in J}\boldc_j $ for some~$J\subseteq [F]$, where~$\boldc_j\in\cC_j$ for all~$j\in J$. Each~$\boldc_j$ belongs to~$\FtwoSpan(\cB')$; otherwise, some additional elements from~$\cB_j$ can be added to~$\cB'$, in contradiction to it being maximal.
	
	The above argument implies a simple algorithm for solving~$\oplus$-recovery, as follows. Suppose that some bases $\cB_1,\ldots,\cB_F$ of~$\cC_1,\ldots,\cC_F$, respectively, are given as input (as rows of the respective generator matrices, see Remark~\ref{remark:bases}). Given these bases, find~$\cB'$ by sequentially adding elements from~$\cup_{i=1}^F\cB_i$ one-by-one, while skipping elements that belong to the span of the elements added thus far. Then, arrange the elements of~$\cB'$ as rows of a matrix~$B$, and given~$\boldc$ solve the equation~$\boldx\odot B=\boldc$. According to the previous argument, this equation has a solution if and only if the required~$\boldc_1,\ldots,\boldc_F$ exist. If exist, each~$\boldc_i$ can be found by linearly combining the elements of~$\cB'\cap\cB_i$ according to the respective coefficients in the solution~$\boldx$.
\end{proof}

The complexity of the above algorithm can be seen as at most cubic in~$\Delta$ and linear in~$n$, since finding~$\cB'$ may require~$\Delta$ Gaussian elimination algorithms, where~$\Delta=\sum_{i=1}^F\dim_{\bF_2}(\cC_i)$, each with complexity at most~$O(\Delta^2 n)$. Then, solving the equation~$\boldx B=\boldc$ takes at most~$O(\Delta^2 n)$ as well. The overall complexity is therefore~$O(\Delta^3n)$.

\begin{remark}\label{remark:XORuniqueness}
	In general, the result of~$\oplus$-recovery might not be unique. For example, this would be the case for~$F=2$ with subcodes~$\cC_1$ and~$\cC_2$ such that~$\cC_1\cap\cC_2=\{\1,\boldv\}$, and~$\boldc=\boldv$. Then, it is clear that~$\boldc_1=\1,\boldc_2=\boldv$ and~$\boldc_1=\boldv,\boldc_2=\1$ are both valid solutions. This is the reason for encoding compositional structures using a product of subcodes (e.g.,~$\cC_1\times\cC_2$ for~$F=2$), which implies that~$\cC_1\cap\cC_2=\{\1\}$.
\end{remark}

\subsection{$\Sigma$-recovery}\label{section:SigmaRecovery}
The~$\Sigma$-recovery problem is formally defined as follows.
\begin{itemize}
	\item \textbf{Input}: A code~$\cC$, a vector $\bolds=\sum_{i\in[s]}\boldc_i$ for some~$\boldc_i$'s in~$\cC$, and the integer~$s$.
	\item \textbf{Goal}: Find~$\boldc_1,\ldots,\boldc_s$.
\end{itemize}
That is, a~$\Sigma$-recovery algorithm receives as input a vector obtained by bundling a set~$\cS$ of~$s$ codewords from~$\cC$, and finds these codewords. We assume that~$s$ is at most the information capacity of~$\cC$, i.e., if~$\cC$ is~$\epsilon$-balanced then~$s<\frac{1}{2}+\frac{1}{4\epsilon}$ (see Section~\ref{section:finitesets}). Otherwise, a solution need not be unique, and the respective vector~$\bolds$ is not guaranteed to provide exact recall. 

Clearly, the efficiency of a~$\Sigma$-recovery algorithm depends on the representation of~$\cC$; we assume here that~$\cC$ is linear and is represented by its parity check matrix (see Section~\ref{section:preliminariesCodes}). As mentioned earlier, $\Sigma$-recovery can be solved via exhaustive search, i.e., one can traverse all codewords~$\boldc\in\cC$ (e.g., by traversing~$\boldx G$ for all~$\boldx\in\bF_2^{\dim\cC}$) and test if~$\boldc\in\cS$ by following the description in Section~\ref{section:examples}; this requires~$|\cC|$ many evaluations. One can heuristically address~$\Sigma$-recovery using resonator networks, and yet, this essentially amounts to exhaustive search; the full details of this approach are given in Section~\ref{section:experiments}. %We propose a more efficient and provably correct method below.

Our algorithm is based on \textit{reducing the size} of the search space using the linear properties of~$\cC$. In particular, by observing the input vector~$\bolds$ one can deduce several properties about the set of codewords~$\cS$. Combining these properties with the fact that any~$\boldc\in\cC$ is a solution to the equation~$H\boldx^\intercal=\1$, the size of the search space can be reduced dramatically.

Before presenting it formally, we exemplify the intuition behind the algorithm. Denote the input vector by~$\bolds=(s_1,\ldots,s_n)$, where each~$s_i$ is an integer, and suppose that~$s_i=s$ for some~$i\in[n]$. Then, one can infer that the~$i$'th entry of \textit{every} codeword in~$\cS$ equals~$1$; similarly, when~$s_i=-s$ one can infer that the~$i$'th entry of every codeword in~$\cS$ equals~$-1$. As a result, all codewords in~$\cS$ are a solution to the equation~$H\boldx'^\intercal=\1$, where~$\boldx'=(x_1',\ldots,x_n')$ is a vector of variables and constants such that~$x_i'=1$ if~$s_i=s$, $x_i'=-1$ if~$s_i=-s$, and otherwise~$x_i'$ is a variable. Namely, we consider the equation~$H\boldx^\intercal=\1$, and fix variables~$x_i$ to constants whenever we are certain that all the~$\boldc_i$'s agree in the~$i$'th entry. Due to the linear structure of~$\cC$, it will be shown that this results in a potentially exponential reduction of the size of the search space---if~$\ell$ entries are fixed it will be shown that a known coding-theoretic bound on the size of the search space shrinks by a multiplicative factor of~$2^{\Theta(\ell)}$ (referring to~$\epsilon$ as a constant).

Formally, for~$i\in\{0,1,\ldots,s\}$ let~$A_i\subseteq[n]$ be the set of indices of~$\bolds$ in which the~$\boldc_i$'s have precisely~$i$-many~$-1$ entries, i.e.,~$s_j=s-2i$ for every~$j\in A_i$. Also, denote~$\boldc_i=(c_{i,1},\ldots,c_{i,n})$ for every~$i\in[s]$. 

\begin{proposition}\label{proposition:A0As}
	For every~$j\in A_0$ we have~$c_{i,j}=1$ for all~$i\in[s]$ and for every~$j\in A_s$ we have~$c_{i,j}=-1$ for all~$i\in[s]$. 
\end{proposition}

Moreover, we show that a similar principle holds for other entries of~$\bolds$, which are not necessarily equal to~$\pm s$, and hence the dimension of the search space could be reduced even further. In what follows we use the abbreviated notation~$\floor{s}_i$ to denote the largest integer such that~$i\floor{s}_i<s$. For example,~$\floor{6}_2=2$ since~$2\cdot2<6$ and~$2\cdot 3=6$.

\begin{lemma}\label{lemma:fixedEntries}
	For every~$i\in\{1,\ldots,\floor{s/2}\}$ and every set~$\cB\subseteq A_i$ of size~$\floor{s}_{i}$, there exists~$\boldc_j$ such that~$c_{j,b}=1$ for all~$b\in\cB$. Similarly, for every~$i\in\{\floor{s/2}+1,\ldots,s-1\}$ and every set~$\cD\subseteq A_i$ of size~$\floor{s}_{s-i}$ there exists~$\boldc_j$ such that~$c_{j,d}=-1$ for all~$d\in\cD$. 
\end{lemma}
\begin{proof}
	Let~$i\in\{1,\ldots,\floor{s/2}\}$, and let~$\cB$ be any subset of~$A_i$ of size~$\floor{s}_{i}$. Consider the bits~$\{c_{j,b}\}_{(j,b)\in [s]\times \cB}$ as a matrix~$M$ of size~$s\times \floor{s}_i$, and observe that~$M$ has exactly~$i$ entries with~$-1$ at each column. Therefore, in any set of~$j$ such columns, the number of rows which contain~$-1$ is at most~$ij$. Since~$M$ has~$s$ rows and~$\floor{s}_i$ columns, and since~$i\cdot \floor{s}_i<s$, it follows that there exists a row with no~$-1$'s; this concludes the first part of the claim. The second part is proved by repeating the proof of the first part, and replacing the roles of~$1$'s and~$-1$'s in~$M$. 
\end{proof}

Our algorithm works in a recursive manner. Each recursive step begins by identifying the sets $A_0,A_1,\ldots,A_s$, and defining a linear equation~$H\boldx'=\1$. In this linear equation, all the entries of~$\boldx'$ indexed by~$A_0$ are fixed to~$1$, all the entries indexed by~$A_s$ are fixed to~$-1$, and some~$\floor{s}_i$ entries in some~$A_i$, $i\in\{1,\ldots,s-1\}$ are fixed to either~$1$ or~$-1$ (depending if~$i\le\floor{s/2}$ or not). To maximize the number of fixed entries, and thereby reducing the size of the search space as much as possible, we will find~$i\in\{1,\ldots,s-1\}$ such that~$\min\{|A_i|, \floor{s}_{\min\{i,s-i\}}\}$ is maximum\footnote{The minimum operation in the subscript of~$\floor{s}_{\min\{i,s-i\}}$ is meant to distinguish between the cases~$i\le\floor{s/2}$ and~$i>\floor{s/2}$ (see Lemma~\ref{lemma:fixedEntries}). The minimum operation between~$\floor{s}_{\min\{i,s-i\}}$ and~$|A_i|$ is for edge cases where~$|A_i|$ is smaller than~$\floor{s}_{\min\{i,s-i\}}$, in which all entries of~$|A_i|$ could be fixed.}. 

We proceed to solve the equation system~$H\boldx'=\1$, i.e., using standard linear algebra over~$\bF_2$ we find a solution space of the form~$\FtwoSpan\{\boldb_1,\boldb_2,\ldots\}\oplus\boldl$ for some basis vectors~$\boldb_i$ and some shift vector~$\boldl$ (notice that this entire solution space is contained in the code~$\cC$). Then, we traverse the solution space in search of a codeword~$\boldc\in\cS$; the truth value of~$\boldc\in\cS$ can be decided by computing~$\bolds\boldc^\intercal$, see Section~\ref{section:examples}. Once one of the~$\boldc_i$'s is found successfully, it is reported to the user, and a recursive call for the same algorithm is be made with the parameters~$\bolds-\boldc_i$, $s-1$, and the same parity check matrix~$H$ which represents~$\cC$. In the basis of the recursion, where~$s=1$, we report~$\bolds$ and conclude. These steps are presented in Algorithm~\ref{algorithm:SigmaRecovery}. We now turn to prove the correctness of Algorithm~\ref{algorithm:SigmaRecovery}, and present several statements regarding its performance. 

\begin{algorithm}
	\caption{$\Sigma$-recovery of linear codes}\label{algorithm:SigmaRecovery}
	\begin{algorithmic}[1]
		\Procedure{$\Sigma$-recovery}{$H,\bolds,s$}\Comment{$\bolds$ is~$s$-bundling of codewords from~$\cC=\ker H$.}
		\State If~$s=1$ \Return $\bolds$.				
		\For{$i=0$ to $s$}
		\State Let~$A_i\subset[n]$ be the set of indices~$j$ such that~$s_j=s-2i$.
		\EndFor
		\State Let~$i_\text{max}=\arg\max_{i\in[n]}\min\{|A_i|, \floor{s}_{\min\{i,s-i\}}\}$.
		\State Let~$g_\text{max}=\min\{|A_{i_\text{max}}|, \floor{s}_{\min\{i_\text{max},s-i_\text{max}\}}\}$, and pick any~$\cB\subseteq A_{i_\text{max}}$ of size~$g_\text{max}$.
		\State Let~$\boldx'=(x_1',\ldots,x_n')$ be a vector of variables and constants such that
		\begin{align*}
			x_i'=\begin{cases}
				\phantom{-}1  & \mbox{if }i\in A_0\\
				-1 & \mbox{if }i\in A_s\\
				\phantom{-}1 & \mbox{if }i\in \cB \mbox{ and }i_\text{max}\le\floor{s/2}\\
				-1 & \mbox{if }i\in \cB \mbox{ and }i_\text{max}>\floor{s/2}\\
				\text{variable} &\text{otherwise}.
			\end{cases}.
		\end{align*}
		\State Solve~$H\boldx'^\intercal=\1$ over~$\bF_2$, and let~$\cR$ be the set of solutions (an affine subspace).
		\For{$\boldc\in\cR$}
		\If{$\boldc$ is in the set represented by~$\cS$ (see Section~\ref{section:examples})}
		\State Output~$\boldc$, call~$\Sigma\textrm{-RECOVERY}(H,\bolds-\boldc,s-1)$, and \Return
		\EndIf
		\EndFor
		\EndProcedure
	\end{algorithmic}
\end{algorithm}		

\begin{theorem}
	Let~$\cC$ be an~$\epsilon$-balanced code for some~$\epsilon>0$, and let~$\bolds=\sum_{i=1}^s\boldc_i$ where~$\boldc_i\in\cC$ for all~$i\in[s]$, and where~$s<\frac{1}{2}+\frac{1}{4\epsilon}$. Then, Algorithm~\ref{algorithm:SigmaRecovery} returns~$\boldc_1,\ldots,\boldc_s$. 
\end{theorem}
\begin{proof}
	Clearly, the theorem holds if and only if (a) there exists~$j\in[s]$ such that~$\boldc_j\in\cR$ (where~$\cR$ is defined in Algorithm~\ref{algorithm:SigmaRecovery}); and (b) it is possible to identify that $\boldc_j\in\cS$ using the algorithms in Section~\ref{section:examples} (for finite sets). Statement~(a) follows from Proposition~\ref{proposition:A0As} and Lemma~\ref{lemma:fixedEntries}, and Statement~(b) follows from the fact that~$s<\frac{1}{2}+\frac{1}{4\epsilon}$.
\end{proof}

Due to the random structure of the underlying linear code~$\cC$, and the arbitrary choice of codewords $\boldc_1,\ldots,\boldc_s\in\cC$, exact runtime bound for Algorithm~\ref{algorithm:SigmaRecovery} is hard to come by. Nevertheless, it is evident that a certain number~$\ell>0$ of entries of~$\boldx'$ are fixed at each iteration. In what follows it is shown that the extent of the reduction in the size of the search space can be estimated using this parameter~$\ell$. We make two statements in this regard, both of which rely on the following proposition.

\begin{proposition}\label{proposition:R'}
	Suppose~$\ell>0$ entries of~$\boldx'$ are fixed in some recursive step of Algorithm~\ref{algorithm:SigmaRecovery}. Then, the solution space~$\cR$ is a (not-necessarily-linear) subcode of~$\cC$ in which~$\ell$ entries are identical across all~$\boldc\in\cR$. Let~$\cR'$ be the result of deleting all fixed entries of~$\cR$. Then,~$\cR'$ is a code of length~$n-\ell$, size~$|\cR|$, and minimum distance at least that of~$\cC$. 
\end{proposition}
\begin{proof}
	The statement about~$\cR$ are clear from its definition. Since~$\cR\subseteq\cC$, it follows that~$\cR$ has minimum distance at least that of~$\cC$. Since~$\cR'$ is obtained from~$\cR$ by deleting entries in which all codewords of~$\cR$ agree, it follows that no two words in~$\cR$ become identical in~$\cR'$, and hence~$|\cR|=|\cR'|$. From the latter argument it also follows that the minimum Hamming distance of~$\cR'$ is identical to that~$\cR$, which in turn is at least the minimum distance of~$\cC$. 
\end{proof}

A priori, it is not clear that fixing some of the entries of~$\boldx'$, as done in Algorithm~\ref{algorithm:SigmaRecovery}, indeed reduces the size of the search space; this would depend on the rank of the submatrix of~$H$ which consists of columns of~$H$ indexed by the variable entries of~$\boldx'$. However, having understood that~$\cR'$ is a code of size~$|\cR|$, length~$n-\ell$, and minimum distance at least that of~$\cC$, we are able to use known coding-theoretic results to bound the size of~$\cR'$, and consequently the size of the search space~$\cR$.

%%%%% Old sphere-packing argument %%%%%	
%	\begin{lemma}
%		Following the notation of Proposition~\ref{proposition:R'}, we have that
%		\begin{align*}
%			|\cR|\le\frac{2^{n(1-h(\frac{1-2\epsilon}{4}))+o_n(1)}}{2^{\ell(1-h(\frac{1-2\epsilon}{4}))+o_n(1)}},
%		\end{align*}
%	where~$h$ is the binary entropy function (i.e., $h(z)=-z\log(z)-(1-z)\log(1-z)$ for~$0\le z\le 1$), and~$o_n(1)$ is an expression which goes to~$0$ and~$n$ tends to infinity.
%	\end{lemma}
%	\begin{proof}
%		We apply the well-known \textit{sphere-packing bound} in its asymptotic version~\cite[Thm.~4.9]{roth2006introduction} over the code~$\cR'$, whose length is~$n-\ell$, and minimum distance is at least~$(1/2-\epsilon)n$ by Proposition~\ref{proposition:R'}. We get
%		\begin{align*}
%			\frac{\log|\cR'|}{n-\ell}\le 1-h\left(\frac{1-2\epsilon}{4}\right)+o_n(1),
%		\end{align*}
%	from which the claim follows since~$|\cR|=|\cR'|$.
%	\end{proof}
%%%%% END %%%%%

\begin{lemma}\label{lemma:MRRW}
	Following the notation of Proposition~\ref{proposition:R'}, we have that
	\begin{align*}
		|\cR|\le
		\frac
		{2^{n   (h(1/2-\tfrac{n}{n-\ell}\sqrt{(1/2-\epsilon)(1/2+\epsilon)}))+o_n(1)}}
		{2^{\ell(h(1/2-\tfrac{n}{n-\ell}\sqrt{(1/2-\epsilon)(1/2+\epsilon)}))+o_n(1)}}
	\end{align*}
	where~$h$ is the binary entropy function (i.e., $h(z)=-z\log(z)-(1-z)\log(1-z)$ for~$0\le z\le 1$), and~$o_n(1)$ is an expression which goes to~$0$ as~$n$ tends to infinity.
\end{lemma}
\begin{proof}
	We apply the \textit{MRRW} bound~\cite[Sec.~4.5]{roth2006introduction} over the code~$\cR'$, whose length is~${n-\ell}$, and minimum distance is at least~$(1/2-\epsilon)n$ by Proposition~\ref{proposition:R'}. We get
	\begin{align*}
		\frac{\log|\cR'|}{n-\ell}\le h(1/2-\tfrac{n}{n-\ell}\sqrt{(1/2-\epsilon)(1/2+\epsilon)})+o_n(1),
	\end{align*}
	from which the claim follows since~$|\cR|=|\cR'|$.
\end{proof}

Lemma~\ref{lemma:MRRW} implies that fixing~$\ell$ entries in a certain step of our~$\Sigma$-recovery algorithm shrinks the MRRW bound on~$|\cR|$ by a factor which is exponential in~$\ell$ (relative to the MRRW bound applied on~$|\cC|$ in its entirety).

When~$\ell$ is large, we are able to provide a stronger bound. This follows from the \textit{Plotkin bound} in coding theory, which is more effective than the MRRW bound in Lemma~\ref{lemma:MRRW} in cases where the minimum distance of~$\cR'$ is at least half its length, i.e., if~$(1/2-\epsilon)n\ge \frac{n-\ell}{2}$.

\begin{lemma}
	Following the notation of Proposition~\ref{proposition:R'}, if~$\ell/n>2\epsilon$, we have that
	\begin{align*}
		|\cR|\le\frac{1-2\epsilon}{\ell/n-2\epsilon}.
	\end{align*}
\end{lemma}

\begin{proof}
	We apply the \textit{Plotkin} bound~\cite[Prob.~4.23]{roth2006introduction} over~$\cR'$ and get
	\begin{align*}
		\frac{(1/2-\epsilon)n}{n-\ell}&\le \frac{1}{2-2/|\cR'|}\\
		\frac{2}{|\cR'|}&\ge \frac{\ell-2\epsilon n}{(1/2-\epsilon)n}\\
		|\cR'|&\le \frac{1-2\epsilon}{\ell/n-2\epsilon},
	\end{align*}
	where the last step follows since~$\ell-2\epsilon n>0$. This concludes the claim since~$|\cR|=|\cR'|$.
\end{proof}

\begin{remark}
	Our experiments show that these bounds are pessimistic, and the our algorithm outperforms exhaustive search often by an order of magnitude. 
\end{remark}

\begin{remark}
	It should be mentioned that a problem which resembles~$\Sigma$-recovery of linear codes was previously studied. The so-called \emph{binary adder channel}~\cite{hughes1996nearly} is a communication model between multiple senders and one receiver, which obtains the sum of the binary messages from the senders, and needs to decompose that sum to its constituent elements. However, in the binary adder channel there is no notion of recall, and therefore traditional code constructions for the binary adder channel are not necessarily incoherent, making them unsuitable for HDC. Furthermore, traditional code constructions for the binary adder channel are not linear. One exception is~\cite{liva2021coding}, which studied the binary adder channel with two senders and linear codes, and presented a decoding algorithm which is a special case of the above~$\Sigma$-recovery algorithm for~$s=2$.
\end{remark}

\section{Experiments}\label{section:experiments}

\subsection{State-of-the-art}
The state-of-the-art method for both $\oplus$-recovery and~$\Sigma$-recovery is resonator networks \cite{frady2020resonator,kent2020resonator}, and its improvements \cite{hersche2023decoding} and \cite{langenegger2023memory}. Although resonator networks are formally designed to address the~$\oplus$-recovery problem (see~\cite[Sec.~2]{kent2020resonator}), they can be extended to perform~$\Sigma$-recovery via \textit{explaining away}. That is, it was demonstrated experimentally (e.g.,~\cite[Fig.~4]{frady2020resonator}) that if one applies resonator networks over a bundle of bound vectors, the network tends to converge to one of the bound vectors, which can then be subtracted, or ``explained away,'' from the bundle. In detail, in resonator networks one maintains an estimated vector for each one of the input subcodes. After convergence, the estimated vectors are bound, the resulting bound vector is given as output, and then subtracted from the input vector. After subtraction, resonator networks can be applied again, until one if left with the zero vector (or close to it). %Each vector that was explained away through this process is presumably a bound vector (with no bundling) and hence resonator networks can be applied in a straightforward manner. The author is not aware of any theoretical justification for this phenomenon. 

Resonator networks, however, are not intended to perform~$\Sigma$-recovery from a single code (for~$F=1$ many subcodes, the formula for computing the estimation vector~$\hat{\boldo}_1$ in~\cite[Eq.~(3.5)]{kent2020resonator} does not apply). Hypothetically, for~$F=1$ one may set the update rule in resonator networks to~$\operatorname{sign}(X_1X_1^\intercal\bolds^\intercal)$ (where~$X_1$ is a matrix containing all codewords as columns), which essentially reduces to exhaustive search---since~$X_1^\intercal\bolds^\intercal$ is the inner product between the input vector~$\bolds$ and every codeword in the code, one can conclude the~$\Sigma$-recovery algorithm at this point according to the inner product values, and further computation using resonator networks is unnecessary. %Therefore, we structure our experiments as follows.

\subsection{$\oplus$-recovery}
For~$\oplus$ recovery we compare against the recent benchmark resonator networks implementation \texttt{torchhd} by~\cite{heddes2023torchhd}. For parameters~$n,k,F$ we randomly choose~$F$ subcodes~$\cC_1,\ldots,\cC_F$ of length~$n$ and dimension~$k$, randomly choose\footnote{A uniformly random codeword in a code~$\cC$ of dimension~$k$ is of the form~$\boldx G$, where~$\boldx$ is a uniformly random~$\{\pm1\}$ vector of length~$k$, and~$G$ is a generator matrix for~$\cC$.}~$\boldc_i\in\cC_i$ for every~$i\in[F]$, and compute the bound vector~$\bolds=\bigoplus_{i=1}^F\boldc_i$. We then implement our~$\oplus$-recovery algorithm (Sec.~\ref{section:oplusRecovery}) on~$\bolds$, and compare against the function \texttt{resonator} of~\cite{heddes2023torchhd}. 
We implemented both algorithms using standard Python libraries, including the \texttt{galois} library for~$\bF_2$ computations~\cite{Hostetter_Galois_2020}, and ran them on a standard laptop computer. %\red{\textbf{Our 

In either algorithm, we declare ``success'' if the returned vectors~$\boldc_1',\ldots,\boldc_F'$ satisfy: (1) $\boldc_i'\in\cC_i$ for all~$i\in[F]$; and (2) $\bigoplus_{i\in[F]}\boldc_i'=\bolds$. Notice that in some cases the representation of~$\bolds$ as binding of one vector from each~$\cC_i$ is not unique (see Remark~\ref{remark:XORuniqueness}), and hence one cannot expect exact retrieval of the original~$\boldc_i$'s. These cases, however, are exceedingly rare for most parameter settings.

We note that for some parameter settings the listing of all codewords, which is a precursory step in resonator networks but not in our algorithm, takes a significant amount of time. Therefore, if listing out all codewords takes at least~$10$ times longer than our entire~$\oplus$-recovery algorithm, we consider that resonator networks failed.

Evidently, our~$\oplus$-recovery algorithm overwhlemingly outperforms the implementation of resonator networks by~\cite{heddes2023torchhd} in a wide range of parameters. We tested both algorithms in all combinations of~$n\in\{500,1000,2000\}$,~$k\in\{3,5,7\}$ and~$F\in\{3,4,5\}$. Each parameter setting was repeated~$10$ times. 

As we proved formally, our $\oplus$-recovery algorithm succeeded in~$100\%$ of the instances, where resonator networks achieved~$0\%$ success in most cases (i.e., it either failed or took more than~$10$ times longer than our methods). We report typical results for~$n=500$ in Table~\ref{table:XORn500}. Very similar results were achieved for~$n\in\{1000,2000\}$, with our algorithm concluding in less than~$0.2$ seconds (with~$\approx 0.02$ standard deviation), and resonator networks achieving~$0\%$ success in almost all cases.

\begin{table}[]
	\centering
	\begin{tabular}{|l|ccccccccc|}
		\hline
		$n$       & \multicolumn{9}{c|}{500}                                                                                                                                                                                                                                                                                                                                                                                                                                                                                                                                                                                                                                                                                            \\ \hline
		$k$       & \multicolumn{3}{c|}{3}                                                                                                                                                                                                                       & \multicolumn{3}{c|}{5}                                                                                                                                                                                                                        & \multicolumn{3}{c|}{7}                                                                                                                                                                                               \\ \hline
		$F$       & \multicolumn{1}{c|}{3}                                                        & \multicolumn{1}{c|}{4}                                                       & \multicolumn{1}{c|}{5}                                                        & \multicolumn{1}{c|}{3}                                                        & \multicolumn{1}{c|}{4}                                                        & \multicolumn{1}{c|}{5}                                                        & \multicolumn{1}{c|}{3}                                                       & \multicolumn{1}{c|}{4}                                                       & \multicolumn{1}{c|}{5}                                 \\ \hline
		Ours      & \multicolumn{1}{c|}{\begin{tabular}[c]{@{}c@{}}0.016 \\ (100\%)\end{tabular}} & \multicolumn{1}{c|}{\begin{tabular}[c]{@{}c@{}}0.021\\ (100\%)\end{tabular}} & \multicolumn{1}{c|}{\begin{tabular}[c]{@{}c@{}}0.032 \\ (100\%)\end{tabular}} & \multicolumn{1}{c|}{\begin{tabular}[c]{@{}c@{}}0.026 \\ (100\%)\end{tabular}} & \multicolumn{1}{c|}{\begin{tabular}[c]{@{}c@{}}0.035 \\ (100\%)\end{tabular}} & \multicolumn{1}{c|}{\begin{tabular}[c]{@{}c@{}}0.067 \\ (100\%)\end{tabular}} & \multicolumn{1}{c|}{\begin{tabular}[c]{@{}c@{}}0.040\\ (100\%)\end{tabular}} & \multicolumn{1}{c|}{\begin{tabular}[c]{@{}c@{}}0.083\\ (100\%)\end{tabular}} & \begin{tabular}[c]{@{}c@{}}0.14\\ (100\%)\end{tabular} \\ \hline
		Resonator & \multicolumn{1}{c|}{\begin{tabular}[c]{@{}c@{}}0.17 \\ (40\%)\end{tabular}}   & \multicolumn{1}{l|}{\begin{tabular}[c]{@{}l@{}}0.18\\ (10\%)\end{tabular}}   & \multicolumn{1}{c|}{$\dagger$}                                                & \multicolumn{1}{c|}{$\dagger$}                                                & \multicolumn{1}{c|}{$\dagger$}                                                & \multicolumn{1}{c|}{$\dagger$}                                                & \multicolumn{1}{c|}{$\dagger$}                                               & \multicolumn{1}{c|}{$\dagger$}                                               & $\dagger$                                              \\ \hline
	\end{tabular}
	\vspace{0.3cm}
	\caption{Run time and success rate comparison between our $\oplus$-recovery algorithm (Section~\ref{section:oplusRecovery}) and resonator networks's benchmark implementation~\cite{heddes2023torchhd}. A~$\dagger$ denotes~$0\%$ success rate out of~$10$ experiments.}\label{table:XORn500}
\end{table}

\subsection{$\Sigma$-recovery}
For~$\Sigma$-recovery we compare against exhaustive search. For parameters~$n,k,s$, we randomly choose a linear code~$\cC$ of length~$n$ and dimension~$k$, and let~$\bolds$ be the bundling of~$s$ many vectors from~$\cC$ that are chosen uniformly at random. Then, we apply our~$\Sigma$-recovery algorithm (Sec.~\ref{section:SigmaRecovery}) and an exhaustive search (which merely iterates over all of~$\cC$ and reports all codewords~$\boldc\in\cC$ such that~$\boldc\bolds^\intercal$ is sufficiently large). In either algorithm, we declare ``success'' only if the returned codewords are identical to those that were initially bundled to create~$\bolds$. 
Similarly, we used standard Python libraries, including~\texttt{galois} as mentioned earlier. The results for various parameters are reported in Table.~\ref{table:Sigma}. 

Evidently, our $\Sigma$-recovery algorithm outperforms exhaustive search in the majority of cases, with roughly an order-of-magnitude improvement, e.g., for the parameters~$(n,k,s)\in\{(200,10,3),(500,12,5),(1000,18,9)\}$. Additionally, several intriguing phenomena are apparent. 

First, exhaustive search retains roughly similar run-times for a given~$k$ and increasing values of~$s$. This is to be expected since exhaustive search iterates (almost) all~$2^k$ codewords regardless of~$s$. In contrast, our~$\Sigma$-recovery slows down by increasing~$s$ (see e.g.~$(n,k)=(500,10)$); this is to be expected as well since our~$\Sigma$-recovery algorithm has fewer entries to fix for large values of~$s$, and hence the reduction in search space size is less drastic (see, e.g., Lemma~\ref{lemma:MRRW}). These differences become less prominent as~$n$ increases, and yet, they indicate that our algorithm performs well when~$s$ is small in comparison to the information capacity of the code (i.e.,~$s<1/2+1/4\epsilon$ for an~$\epsilon$-balanced code, see e.g., Theorem~\ref{theorem:returnCnC}). 

Second, the above property of exhaustive search seems to come at a price. When~$s$ is pushed beyond the information capacity, exhaustive search tends to fail more often than our algorithm. We do not have a rigorous explanation for this phenomenon, and yet we postulate that traversing the restricted solution space~$\cR$ in our algorithm, rather than the entire code, contributes to making less errors since it decreases the chance to encounter erroneous codewords (i.e., codewords~$\boldc\in\cC$ for which~$\bolds\boldc^\intercal$ is high, while~$\boldc$ is not one of the superposed codewords in~$\bolds$).

Third, our~$\Sigma$-recovery algorithm tends to variability. Upon closer scrutiny, one finds that this is due to a dichotomy in run-times. For example, for~$(n,k,s)=(1000,14,9)$, run-times are rather consistently partitioned to~$\approx 30$\textsubscript{sec} and~$\approx 250$\textsubscript{sec}. We postulate that this phenomenon follows from concentration of the superimposed vectors, and yet a more rigorous explanation is left for future studies.

\begin{table}[]
	\centering
	\begin{tabular}{|c|c|c|cc|}
		\hline
		\multirow{2}{*}{$n$}    & \multirow{2}{*}{$k$}  & \multirow{2}{*}{$s$} & \multicolumn{2}{c|}{Success/overall, Avg. run-time (std)}                                \\ \cline{4-5} 
		&                       &                      & \multicolumn{1}{c|}{Our~$\Sigma$-recovery (Section~\ref{section:SigmaRecovery})}                                                  & Exhaustive search                                                \\ \hline
		\multirow{4}{*}{$200$}  & \multirow{2}{*}{$8$}  & $3$                  & \multicolumn{1}{c|}{$100/100, 9.7\cdot 10^{-2}$\textsubscript{sec} ($8.2\cdot 10^{-3}$)} & $100/100, 0.52$\textsubscript{sec} ($22.3\cdot 10^{-3}$) \\ \cline{3-5} 
		&                       & $5$                  & \multicolumn{1}{c|}{$100/100, 0.25$\textsubscript{sec} ($1.4\cdot 10^{-2}$)}             & $94/100, 0.52$\textsubscript{sec} ($20.4\cdot 10^{-3}$)     \\ \cline{2-5} 
		& \multirow{2}{*}{$10$} & $3$                  & \multicolumn{1}{c|}{$100/100, 0.095$\textsubscript{sec} ($8.08\cdot 10^{-3}$)}           & $100/100, 2.08$\textsubscript{sec} ($46.4\cdot 10^{-3}$)            \\ \cline{3-5} 
		&                       & $5$                  & \multicolumn{1}{c|}{$100/100, 0.32$\textsubscript{sec} ($0.53$)}                         & $81/100, 2.08$\textsubscript{sec} ($37.7\cdot 10^{-3}$)             \\ \hline
		\multirow{6}{*}{$500$}  & \multirow{2}{*}{$10$} & $5$                  & \multicolumn{1}{c|}{$100/100, 1.36$\textsubscript{sec} ($4.98\cdot 10^{-2}$)}            & $100/100, 5.09$\textsubscript{sec} ($1.37\cdot 10^{-1}$)            \\ \cline{3-5} 
		&                       & $7$                  & \multicolumn{1}{c|}{$100/100, 5.47$\textsubscript{sec} ($8.38$)}                         & $99/100, 5.12$\textsubscript{sec} ($9.7\cdot 10^{-2}$)              \\ \cline{2-5} 
		& \multirow{2}{*}{$12$} & $5$                  & \multicolumn{1}{c|}{$100/100, 1.34$\textsubscript{sec} ($4.92\cdot 10^{-2}$)}            & $100/100, 20.3$\textsubscript{sec} ($3.37 \cdot 10^{-1}$)           \\ \cline{3-5} 
		&                       & $7$                  & \multicolumn{1}{c|}{$100/100, 10.9$\textsubscript{sec} ($12$)}                           & $100/100, 20.3$\textsubscript{sec} ($2.67\cdot 10^{-1}$)            \\ \cline{2-5} 
		& \multirow{2}{*}{$14$} & $7$                  & \multicolumn{1}{c|}{$100/100, 20.1$\textsubscript{sec} ($14$)}                           & $97/100, 81.2$\textsubscript{sec} ($0.98$)                          \\ \cline{3-5} 
		&                       & $9$                  & \multicolumn{1}{c|}{$98/100, 75.4$\textsubscript{sec} ($16$)}                            & $83/100, 80$\textsubscript{sec} ($0.89$)                            \\ \hline
		\multirow{6}{*}{$1000$} & \multirow{2}{*}{$14$} & $7$                  & \multicolumn{1}{c|}{$10/10, 25.6$\textsubscript{sec} ($0.54$)}                         & $10/10, 160.3$\textsubscript{sec} ($1.4$)                         \\ \cline{3-5} 
		&                       & $9$                  & \multicolumn{1}{c|}{$10/10, 145.03$\textsubscript{sec} ($109$)}                        & $10/10, 159.08$\textsubscript{sec} ($1.4$)                        \\ \cline{2-5} 
		& \multirow{2}{*}{$16$} & $7$                  & \multicolumn{1}{c|}{$10/10, 58.05$\textsubscript{sec} ($49.79$)}                         & $10/10, 1149.16$\textsubscript{sec} ($131$)                        \\ \cline{3-5} 
		&                       & $9$                  & \multicolumn{1}{c|}{$10/10, 406.48$\textsubscript{sec} ($149$)}                         & $10/10, 1087.32$\textsubscript{sec} ($54.33$)                        \\ \cline{2-5} 
		& \multirow{2}{*}{$18$} & $9$                  & \multicolumn{1}{c|}{$10/10, 272.9$\textsubscript{sec} ($68.9$)}                        & $10/10, 2569.9$\textsubscript{sec} ($88$)                         \\ \cline{3-5} 
		&                       & $11$                 & \multicolumn{1}{c|}{$10/10, 500.9$\textsubscript{sec} ($58.2$)}                        & $10/10, 2526.3$\textsubscript{sec} ($34.3$)                       \\ \hline
	\end{tabular}
	\vspace{0.3cm}
\caption{Comparison of our~$\Sigma$-recovery algorithm (Section~\ref{section:SigmaRecovery}) against exhaustive search.}\label{table:Sigma}
\end{table}

\subsection{Retrieval capacity}
Finally, while our focus in this paper is on efficient and exact recovery algorithms, one may also wonder if using linear codes affects the retrieval capacity. To this end, we repeat the experimental setting of~\cite[Table~1]{hersche2023decoding}, as follows. Motivated by the text classification task of Amazon-12K~\cite{ganesan2021learning}, in \cite{hersche2023decoding} they consider~$F=2$ codes containing~$110$ codewords each and length~$n=256$. Then, for successive value~$N=1,2,\ldots$, they randomly choose~$N$ pairs of codewords, one from each code, bind each pair, and bundle the~$N$ resulting bound vectors. They then test their methods in retrieving all~$2N$ vectors from the bundle.

To apply linear codes in the above setting, we randomly choose a~$[256,14]_2$ code~$\cC$, and fix two subcodes~$\cC_1,\cC_2$ of dimension~$7$ so that~$\cC=\cC_1\times\cC_2$ (a code of dimension~$7$ has~$2^7=128$ codewords, and hence suffices to store~$110$ values). We choose~$N$ pairs~$\{(\boldc_{i,1},\boldc_{i,2})\}_{i=1}^N$ uniformly at random, where~$\boldc_{i,1}\in\cC_1$ and~$\boldc_{i,2}\in\cC_2$ for all~$i\in[N]$, and define the vector~$\bolds=\sum_{i=1}^N \boldc_{i,1}\oplus\boldc_{i,2}$. Since~$\boldc_{i,1}\oplus\boldc_{i,2}\in\cC$ for all~$i\in[N]$, a~$\Sigma$-recovery algorithm for~$\cC$ factorizes~$\bolds$ to the individual bound vectors~$\{\boldc_{i,1}\oplus\boldc_{i,2}\}_{i=1}^N$. Then, our~$\oplus$-recovery algorithm is applied to unbind each~$\boldc_{i,1}\oplus\boldc_{i,2}$.

\begin{table}[]
	\centering
	\begin{tabular}{|c|c|c|c|c|c|}
		\hline
		$s$           & $4$       & $5$       & $6$        \\ \hline
		Avg. run time & $0.31$\textsubscript{sec}& $0.96$\textsubscript{sec}& $7.06$\textsubscript{sec}\\ \hline
		Success rate  & $100\%$    & $99.9\%$   & $97.1\%$   \\ \hline
	\end{tabular}
	\vspace{0.3cm}
\caption{Run times and success rate over~$1000$ experiments of~$\Sigma$-recovery with parameters~$(n,k)=(256,14)$, which recreates~\cite[Table~1]{hersche2023decoding}. Additional~$\oplus$ recovery with parameters~$(n,k,F)=(256,7,2)$ took on average~$0.02$\textsubscript{sec}and succeeded~$100\%$ of the time.}\label{table:Hersche}
\end{table}

It is reported in~\cite{hersche2023decoding} that resonator networks achieved~$>99\%$ success in retrieving all~$\boldc_{i,j}$'s only for~$N=1$, thereby challenging their explaining-away functionality. A method called ``sequential decoding'' increased that number to~$5$, and a method called ``mixed decoding,'' which conducts several explaining-away instances in parallel and proceeds with the best one, raised this number to~$8$. In our experiments~(Table~\ref{table:Hersche}), for increasing values of~$N$ we measured the average time and the success rate of our~$\Sigma$-recovery algorithm across~$1000$ experiments. For all values of~$N$, further decomposition of each of the resulting~$N$ bound vectors was done using our~$\oplus$-recovery algorithm, that concluded successfully in~$100\%$ of the cases with average running time of~$\approx 0.02$\textsubscript{sec} for each bound pair. As can be seen in Table~\ref{table:Hersche}, the retrieval capacity of our methods seems to be~$s=5$ (with~$s=6$ coming close). It is an interesting direction for future work to see if similar ``mixed decoding'' ideas can increase our retrieval capacity further.

\section{Summary and Discussion}\label{section:discussion}
This paper proposes the use of linear codes in hyperdimensional computing, which are defined via Boolean linear algebra. These codes, when chosen at random, retain favorable incoherence properties relative to not-necessarily-linear ones, and hence can be safely used for HDC at no loss of performance. Yet, their algebraic structure provides an \textit{almost-instantaneous} solution to the~$\oplus$-recovery problem, and can be used to drastically accelerate exhaustive search methods in~$\Sigma$-recovery. 

On top of these benefits, linear codes also suggest efficient encoding and exponentially smaller storage space. Both of these properties are guaranteed by the generator matrix of the code. Efficient encoding is given by~$\boldx\mapsto\boldx\odot G$, which can be implemented using exclusive-OR operations only. For storage one can only store~$G$ (which contains~$k\times n$ bits), rather than the entire code (which contains~$2^k\times n$ bits), and still be able to determine if a codeword is in the code using simple linear-algebraic operations. 

An additional benefit is the ability to \textit{verify} that a linear code has certain incoherence properties, faster than a non-linear one. In critical applications one would like to make sure that the code chosen at random indeed possesses the required incoherence properties. To verify the incoherence of a non-linear code one has to compute the inner product between any two codewords, resulting in runtime that is \textit{quadratic} in the size of the code. However, due to the equivalence between incoherence and minimum distance (in all codes), and between minimum distance to minimum weight (in linear codes, see Definition~\ref{definition:epsBalanced} and the subsequent discussion), to verify incoherence in linear codes one only needs to verify minimum weight. This results in runtime that is \textit{linear} in the size of the code.

The options for future work are vast, the most glaring of which seems to be studying linear codes for approximate computation (where computations should succeed with high probability), and finding faster~$\Sigma$-recovery algorithms. Moreover, our focus here has been on simple computational tasks such as recall, recovery, and implementation of various data structures; the role of linear codes in HDC for more learning-oriented tasks such as classification, clustering, or regression~\cite{kleyko2018classification} should be addressed. Finally, coding theoretic methods might find use outside the realm of binary values, and applications of coding theoretic ideas to real-valued~\cite{thomas2021theoretical} or complex-valued~\cite{yu2022understanding} HDC remains to be studied.

\bibliographystyle{ieeetr}
\bibliography{ref.bib}
\appendix
\section{Appendix.}

\subsection{Representing items using (few) bits}\label{section:representation}
Some of the examples in Section~\ref{section:examples} (sets, vectors, etc.) require that items be represented using a small number of bits, prior to being encoded to HD vectors using a generator matrix. Below we discuss a few potential approaches.

\noindent\textbf{Using the item's ``natural'' representation}. Most conceivable data types have some form of ``natural'' representation (e.g., ASCII code for text, or RGB for pixels). A wasteful yet simple method for representing items with bits is using these naturally occurring representations. For example, to store an English sentence as a set of words one can use the $7$-bit ASCII representation of the letters in those words. By restricting the supported length of words to, say, $20$~English letters, every word is naturally associated with~$140$ bits, and hence a code of dimension~$k=140$ suffices. This is a highly redundant representation since a vast majority of the~$2^{140}$ binary vectors of length~$140$ would not be associated with any existing English word.

\noindent\textbf{Indexing.} If an indexed dictionary of all possible items is available, one can use the item's index in that dictionary as its binary representation. For example, the Oxford English dictionary consists of~$171,476$ words in current use, which can be indexed using~$18$ bits. This is drastically more compact than the previous option; for example, the experiments at the bottom of Table~\ref{table:Sigma} show that using this~$18$-bit representation with~$1000$ bit codewords can support efficient~$\Sigma$-recovery of sets of size~$9$ and~$11$ English words. This is still somewhat wasteful; one can use fewer than~$18$ bits using compression algorithms, as explained next.

\noindent\textbf{Compression algorithms.} A classic result by Claude Shannon asserts that there exist encoding algorithms for compressing English words to~$11.82$ bits on average, which was later improved~$9.8$ bits on average~\cite{grignetti1964note}. This is a special case the famous \textit{source coding theorem}~\cite{cover1999elements}, which states that any information source~$X$ (seen as a random variable) could be compressed down to~$H(X)$ many bits on average, where~$H$ is information entropy. Optimal compression algorithms exist, such as the Lempel-Ziv algorithm~\cite{ziv1977universal}. Therefore, by ignoring rare words, one can reduce the representation size to below~$18$.

\noindent\textbf{Random.} If neither dictionary nor compression algorithms are available in some setting, one can resort to random coding. Specifically, one would choose random codewords of length~$k$, where~$k$ is such that the universe of all relevant items is of size at most~$2^k$. For example, one can map each encountered English word to a uniformly random~$18$-bit string. Notice that this method would require some bookkeeping for memorizing which word was mapped to which~$18$-bit string. Yet, we emphasize that even this method is far more efficient than memorizing all HD vectors assigned so far, since HD vectors are normally much longer than~$18$ bits in any domain known to the author.

\subsection{Encoding without choosing a code a priori.}\label{section:onthefly}
For ease of analysis, in the present paper it is assumed that the number of elements is known a priori, and one can choose a random~$[n,k]_2$ code~$\cC$ so that the number of elements is at most~$2^k$. In cases that the number of elements is not known a priori, one can choose the code~$\cC$ ``on the fly,'' i.e., in a gradual manner, as follows:
\begin{itemize}
	\item The first element to be encoded is mapped to the HD vector~$\1$. 
	\item The second element is mapped a uniformly random vector~$\boldv_1$.
	\item The third element is mapped a uniformly random vector~$\boldv_2$.
	\item The fourth element is mapped to the vector~$\boldv_1\oplus\boldv_2$.
	\item The fifth element is mapped to a uniformly random vector~$\boldv_3$.
	\item The sixth to eighth vectors are mapped to different~$\bF_2$-linear combinations of~$\boldv_1,\boldv_2,\boldv_3$ that were not already used.
	\item The ninth element is mapped to a uniformly random vector~$\boldv_4$.
	\item The tenth to sixteenth vectors are mapped different~$\bF_2$-linear combinations of~$\boldv_1,\boldv_2,\boldv_3, \boldv_4$ that were not already used.
	\item $\ldots$.
\end{itemize}
More generally, the second element is mapped to a uniformly random vector~$\boldv_1$, and for~$j\ge 1$, the~$(2^j+1)$'th element is mapped to a uniformly random vector~$\boldv_{j+1}$. All other elements are mapped to unique~$\bF_2$-linear combinations of the~$\boldv_i$'s chosen so far. In a way, one can think of the~$\boldv_i$'s as rows of a generator matrix~$G$, whose span is the eventual code~$\cC$ which we use. Alternatively, one can think of the above procedure as on-the-fly indexing, which is followed by encoding using a generator matrix. That is, each item is mapped to an index of length~$\ceil {\log(s)}$, where~$s$ is the number of items encountered so far. Whenever the number of items exceeds an integer power of two, all previous indices are appended with~$1$ (i.e., the zero element of~$\bF_2$ in its~$\{\pm1\}$ representation), and each new item is mapped to indices which end with~$-1$ (i.e., the nonzero element of~$\bF_2$). At any point of this process, an index~$\boldi$ is encoded to the codeword~$\boldi G$, where~$G$ is a matrix whose rows are the~$\boldv_i$'s chosen so far.
	
\end{document}